\documentclass[a4paper]{lipics-v2021}

\usepackage{amsmath}
\usepackage{amsthm}
\usepackage[bibliography=common]{apxproof}

\newcommand{\sched}{\sigma}
\newcommand{\plan}{\mathit{intent}}
\newcommand{\noconflict}{\mathcal{O}}
\newcommand{\Agents}{\mathit{Ag}}
\newcommand{\I}{\mathcal{I}}
\newcommand{\Nat}{\mathbb{N}}
\newcommand{\adversary}{\alpha}
\newcommand{\arrival}{\tau}

\newcommand{\F}{\mathcal{F}}
\newcommand{\R}{\mathcal{R}}

\newcommand{\histories}{\mathcal{H}}
\newcommand{\nofail}{\mathit{NF}}
\newcommand{\CR}{\mathit{CR}}
\newcommand{\SO}{\mathit{SO}}

\newcommand{\GO}{\mathit{GO}}

\newcommand{\stg}{\mathit{safe\text{-}to\text{-}go}}
\newcommand{\pos}{\mathit{pos}}
\newcommand{\lane}{\mathit{lane}}
\newcommand{\go}{\mathtt{go}}
\newcommand{\envstates}{L_e}
\newcommand{\sensormodel}{\mathcal{S}}
\newcommand{\Time}{\mathit{time}}

\usepackage{mathrsfs}
\newcommand{\sense}{\mathscr{S}}
\newcommand{\sensorreading}{\Sigma}
\newcommand{\memory}{\mathit{Mem}}
\newcommand{\initmemory}{\mathit{Mem}^{init}}
\newcommand{\msg}{m}

\newcommand{\gotime}{\mathit{gotime}}
\newcommand{\front}{\mathit{front}}
\newcommand{\going}{\mathit{going}}

\newcommand{\kbp}{\mathbf{P}}
\newcommand{\noop}{\mathtt{noop}}

\newcommand{\exchange}{\mathcal{E}}
\newcommand{\rimp}{\Rightarrow}
\newcommand{\dimp}{\Leftrightarrow}
\newcommand{\Lanes}{\mathcal{L}}
\newcommand{\Lanesin}{\mathcal{L}_\mathit{in}}
\newcommand{\Lanesout}{\mathcal{L}_\mathit{out}}
\newcommand{\lead}{\mathit{next}}

\usepackage{xcolor}

\newcommand{\commentout}[1]{}
\usepackage[ruled,vlined]{algorithm2e}
\newenvironment{program}[1][ht]
  { %
   \renewcommand{\thealgocf}{} %
   \begin{algorithm}[#1]
  }{\end{algorithm}}

\newtheoremrep{proposition}{Proposition}
\newtheoremrep{theorem}{Theorem}
\newtheoremrep{lemma}{Lemma}
\newtheoremrep{corollary}{Corollary}

\newcommand{\citeyear}{\cite}
\newcommand{\canhear}{T}
\newcommand{\first}{first}

\newcommand{\poss}{\mathit{Pos}}

\title{A Knowledge-Based Analysis of Intersection Protocols}

\author{Kaya Alpturer}%
       {Princeton University, USA}%
       {kalpturer@princeton.edu}%
       {https://orcid.org/0000-0003-4843-883X}%
       {Research supported by AFOSR grant FA23862114029.
       Work done in part while the author was studying at Cornell University.}%
\author{Joseph Y. Halpern}%
       {Cornell University, USA}%
       {halpern@cs.cornell.edu}%
       {https://orcid.org/0000-0002-9229-1663}%
       {Research supported in part by AFOSR grant
                FA23862114029, NSF grant FMitF 2319186, ARO grant
                W911NF-22-1-0061, and MURI grant W911NF-19-1-0217.}
\author{Ron van der Meyden}%
       {UNSW Sydney, Australia}%
       {R.VanderMeyden@unsw.edu.au}%
       {https://orcid.org/0000-0002-9243-0571}%
       {The Commonwealth of Australia (represented by the Defence Science and
       Technology Group) supported this research through a Defence Science
       Partnerships agreement.}%

\authorrunning{K. Alpturer, J. Halpern, and R. van der Meyden}

\Copyright{Kaya Alpturer, Joseph Y. Halpern, and Ron van der Meyden} %


\category{}
\nolinenumbers
\hideLIPIcs

\begin{document}

\maketitle

\begin{abstract}
  The increasing wireless communication capabilities of vehicles creates
  opportunities for more efficient intersection management strategies. One
  promising approach is the replacement of traffic lights with a system wherein
  vehicles run protocols among themselves to determine right of way. In this
  paper, we define the \emph{intersection problem} to model this scenario
  abstractly, without any assumptions on the specific structure of the
  intersection or a bound on the number of vehicles. Protocols solving the
  intersection problem must guarantee safety (no collisions) and liveness (every
  vehicle eventually goes through). In addition, we would like these protocols
  to satisfy various optimality criteria, some of which turn out to be
  achievable only in a subset of the contexts. In particular, we show a partial
  equivalence between eliminating unnecessary waiting, a criterion of interest
  in the distributed mutual-exclusion literature, and a notion of optimality
  that we define called \emph{lexicographical optimality}. We then introduce a
  framework to design protocols for the intersection problem by converting an
  \emph{intersection policy}, which is based on a global view of the
  intersection, to a protocol that can be run by the vehicles through the use of
  knowledge-based programs. Our protocols are shown to guarantee safety and
  liveness while also being optimal under sufficient conditions on the context.
  Finally, we investigate protocols in the presence of faulty vehicles that
  experience communication failures and older vehicles with limited
  communication capabilities.
  We show that intersection protocols can be made safe, live and optimal
  even in the presence of faulty behavior.
\end{abstract}

\section{Introduction}
Traffic lights can slow down traffic significantly, due to their lack of
responsiveness to real-time traffic. If
vehicles can communicate with each
other
(which is already quite feasible with today's
wireless
technology), the door is open for improved protocols, where vehicles can
determine right of way  among themselves, depending on traffic
conditions, and thereby significantly increase throughput at an intersection.
In this paper, we formally define the
\emph{intersection problem}: we assume that agents can communicate
with each other via radio broadcasts, and design protocols that take
advantage of this communication to allow
agents to go through the intersection while satisfying \emph{safety}
(no collisions)
and \emph{liveness} (every vehicle
eventually goes through).
In addition, we consider \emph{optimal} protocols, which means,
roughly speaking, that the protocol allows as many vehicles as
possible to go through the intersection at any given time.
Finally, we consider the extent to which we can tolerate communication
failures and (older) vehicles that are not equipped with wireless, so
cannot broadcast messages.  (It turns out that these two
possibilities can be dealt with essentially the same way.)

While the inefficiencies of traffic-light-based intersection
management have long been recognized \cite{ds08},
prior approaches have mainly focused on specific
intersection scenarios \cite{ciscav21,ssp17} or relied on executing
leader-election protocols without considering communication failures
\cite{fvpbk13,ffcvt10}. Furthermore, the protocols have often been
evaluated based
on simulations of specific intersections, rather than being proved
correct \cite{ffcvt10,ciscav21}. Given the
implications of this problem for traffic safety, as well as its potential for
greatly improving energy efficiency and productivity, there is a need for
formal guarantees on both correctness and optimality.

To the best of our knowledge, prior work did not consider optimality,
especially in the presence of various faults.
In designing these protocols,
to the extent possible,
we want them to be robust to a variety of
communication failures, such as contexts with \emph{crash} failures,%
\footnote{We follow the distributed-algorithms literature's
  interpretation of ``crash failure'' here: it is not meant
  to imply a physical collision.}
where an
agent may
fail by ceasing
to participate in the protocol at a given time, and
\emph{omission} failures, where arbitrary messages can fail to be broadcast.

Epistemic logic
has been shown to provide a high-level abstraction that
can be used
to design
distributed
protocols independent of particular assumptions on the communication
environment
and type of failures \cite{FHMV}.
Most analyses of distributed-computing problems that use epistemic
logic have
used full-information protocols to
derive time-optimal algorithms, at the cost of large message size and memory requirements.
Given the
limitations of wireless
networks, it is also desirable to bound the amount of information that
needs to be exchanged between  agents, while still ensuring that the
formal guarantees are still met.
To address this, following \cite{AHM23},
we separate the part of
the protocol that determines
what information is exchanged between the agents, and the part that
determines what action to take based on the agent's information.
Thus, when we consider optimality, we do so with respect to protocols
that
limit information exchange in the same way.

We model the intersection problem as the following scenario.
There is a (possibly infinite) set of agents $\Agents \subseteq \Nat$.
The intersection has
$\ell$ lanes, represented by $\Lanes = \{0,\dots,\ell-1\}$.
The set of lanes is partitioned into a set of lanes
$\Lanesin = \{0, \ldots, k-1\}$, where $1 < k < \ell$
by which vehicles approach the intersection,
and a set of lanes $\Lanesout$
by which they depart from the intersection.
Each lane in $\Lanesin$ has a queue of agents waiting to go through
the intersection;
at each point in time
at most one agent
arrives at each of these queues.
A \emph{move}
through the intersection is represented by a pair
$(l_s, l_t) \in \Lanesin\times\Lanesout.$
Intuitively, executing $(l_s,l_t)$ means that the agent
arrives through lane $l_s$ and departs through lane $l_t$.
The symmetric relation $\noconflict \subseteq (\Lanesin \times \Lanesout)^2$
describes
which moves of the agents are compatible; $((l_s,l_t),(l'_s,l'_t)) \in
  \noconflict$
means that both $(l_s,l_t)$ and $(l'_s,l'_t)$ can be executed in the
same round.
Broadcasts have a limited range, given by  $\rho > 0$.
We assume that, provided there are no failures, all broadcasts sent by
an agent $i$  will be received by all
agents that are within a distance $\rho$  of $i$.
The problem is then to maximize the rate at which cars move through
the intersection while guaranteeing safety (it is never the case that
agents with incompatible moves go through the intersection
simultaneously) and liveness (all agents that arrive at the
intersection eventually move through it).
The problem can be thought of as a generalization of distributed
mutual exclusion, where the intersection is the critical section.
\begin{figure}
  \centering
  \includegraphics[width=0.33\textwidth]{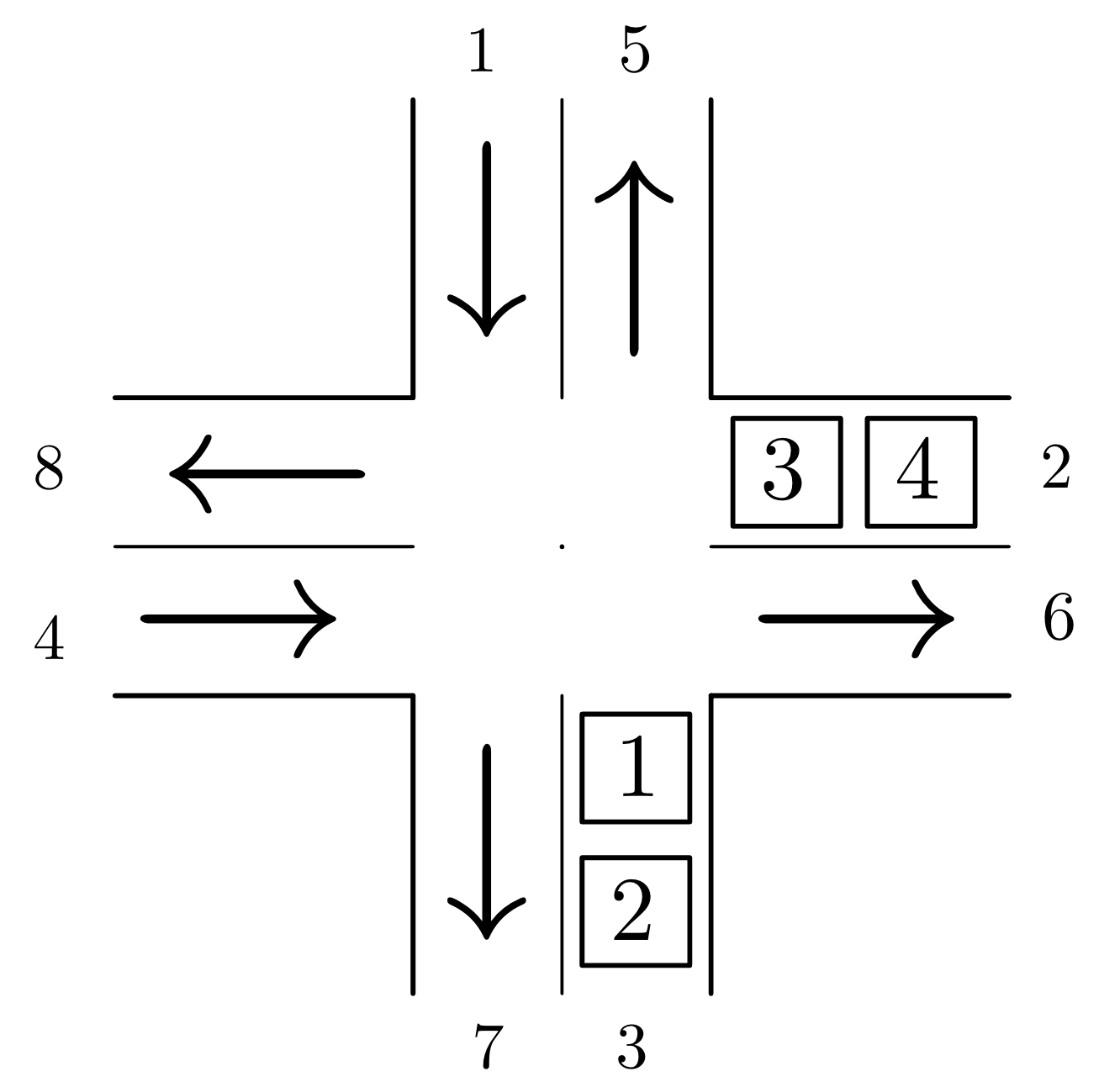}
  \caption{An intersection with $\Lanes = \{1,\dots,8\}$ where
    $\Lanesin = \{1,\dots,4\}$ and $\Lanesout = \{5,\dots,8\}$.
    There are currently 4 agents that have arrived in incoming lanes 2 and 3.}
  \label{fig:model}
\end{figure}

The rest of the paper is organized as follows: In
Section~\ref{sec:rak}, we briefly
review the knowledge-based framework of \cite{FHMV}. In Section~\ref{sec:iep},
we modify the information-exchange model of \cite{AHM23} and introduce the
sensor model. Section~\ref{sec:adv} defines models for the adversary which
determine the arrival schedule of vehicles and communication failures.
Section~\ref{sec:context} combines the information-exchange and the adversary
model, fully specializing the general model of Section~\ref{sec:rak} to
intersections. Section~\ref{sec:unnecessary-waiting} introduces the various
notions of optimality we care about such as eliminating unnecessary waiting and
lexicographical optimality. In Section~\ref{sec:intersection-policies},
intersection policies are introduced as a global view of the intersection.
Section~\ref{sec:optimal} proves a construction that results in an optimal
policy even with failures, and explores applications of the construction
in two limited-information contexts.
Section~\ref{sec:discussion} concludes with a discussion on connections to
distributed mutual exclusion.
We defer most proofs to the appendix.

\section{Reasoning about knowledge} \label{sec:rak}
In order to reason about the knowledge of the vehicles in the intersection
problem, we use the standard runs-and-systems model \cite{FHMV}.
An interpreted system $\I = (\R,\pi)$ consists of
a system $\R$, which is a set of runs, and
an \emph{interpretation} $\pi : \R \times \Nat \to
  \mathcal{P}(\mathit{Prop})$.
Each \emph{run} $r : \Nat \to L_e \times \Pi_{i \in \Agents} L_i$
describes a particular infinite execution of the system
where $r(m)$ is the global state of the system in run $r$ at time $m$.
The global states consist of an environment state drawn from $L_e$ and
local states for each agent $i$ drawn from each $L_i$.
The local state of agent $i$ at point $(r,m)$ is denoted $r_i(m)$.
We call a run and time pair $(r,m)$ a \emph{point}.
The interpretation $\pi$ describes which atomic
propositions hold
at each point in a system $\R$.

We write $\I,(r,m) \models \phi$ if the
formula $\phi$ holds
(is satisfied)
at point $(r,m)$ in interpreted system $\I$.
A formula $\phi$ is \emph{valid in an interpreted system $\I$},
denoted $\I \models \phi$, if $\phi$
holds at
all points in $\I$; the formula $\phi$ is \emph{valid} if it is valid
in all interpreted systems.
Satisfaction of formulas is
inductively defined as follows:
\begin{itemize}
  \item $\I,(r,m) \models p$ iff $p \in \pi(r,m)$.
  \item $\I,(r,m) \models \phi \land \phi'$ iff
        $\I,(r,m) \models \phi$ and $\I,(r,m) \models \phi'$.
  \item $\I,(r,m) \models \neg\phi$ iff $\I,(r,m) \not\models \phi$.
  \item $\I,(r,m) \models K_i\phi$ iff $\I,(r',m') \models \phi$ for all points
        $(r',m')$ such that $r_i(m) = r'_i(m')$.
  \item $\I,(r,m) \models \Diamond \phi$ iff for some $m' \geq m$, $\I,(r,m') \models \phi$.
  \item $\I,(r,m) \models \bigcirc \phi$ iff $\I,(r,m+1) \models \phi$.
\end{itemize}
Agent $i$ \emph{knows} a formula $\phi$ at $(r,m)$ if
$\I,(r,m) \models K_i\phi$.
Intuitively, agent $i$ knows $\phi$ if $\phi$ holds at all points
where
agent $i$ has the same local state. We say that agent $i$ considers the point
$(r',m')$ \emph{possible} at point $(r,m)$ if $r_i(m)=r'_i(m')$.
The relation $\sim_i$ is defined as $(r,m) \sim_i (r',m')$ iff $r_i(m) = r'_i(m')$.
The formula $\bigcirc \phi$ means that $\phi$ holds at the next time, and $\Diamond \phi$ means that $\phi$ holds eventually.
In later sections, we formalize how interpreted
systems for the intersection problem are specified.

\section{Information-exchange protocols} \label{sec:iep}

Our framework for modeling limited information exchange is similar to
that used by Alpturer et al. \citeyear{AHM23} to analyze consensus protocols,
but we make a number of changes due to the differences in our setting.
Here, global states represent not just the result of
messages sent between the agents, but also facts about a changing
external world, from which the agents  obtain sensor readings
(e.g., information about
their
own position and that of nearby vehicles,
from GPS, visual, lidar, or radar sensors).
We modify the definition of  information-exchange protocols
from \cite{AHM23} to accommodate these sensor readings.
Specifically, assume that we are given a set $\envstates$ of
environment states. Define a \emph{sensor model}
for $\envstates$ to  be a collection
of mappings
$\sensormodel = \{\sense_i\}_{i \in \Agents}$, where $\sense_i:\envstates \rightarrow \sensorreading_i$
maps states of the environment to a set $\sensorreading_i$ of
possible
sensor readings for agent $i$.

An \emph{information-exchange protocol} $\exchange$ for agents $\Agents$
and sensor model $\sensormodel$
is
given by the
collection
$\{\exchange_{i}\}_{i\in\Agents}$ consisting of a  local
information-exchange protocol $\exchange_i$ for each agent $i$.  Each
local information-exchange protocol  $\exchange_i$ is a tuple
$\langle L_i, \initmemory_i , A_i, M_i, \mu_i, \delta_i\rangle$, where
\begin{itemize}
  \item
        $L_i = \memory_i \times \sensorreading_i$
        is a set of local states,
        where each local state consists of a memory state from a set
        $\memory_i$ and a sensor reading from
        $\sensorreading_i$;
  \item $\initmemory_i \subseteq \memory_i$ is a set of initial memory states.
        (Typically, there might be a single initial memory state, containing
        information such as the agent's identity.)
  \item $M_i$ is
        the
        set of messages that can be sent by agent $i$;
  \item $\mu_i : L_i \times A_i \times \sensorreading_i \rightarrow M_i \cup \{\bot\}$ is a function
        mapping a local state $s$, an action $a$, and a sensor reading
        $o$ to the message
        to be broadcast (intuitively, $\mu_i(s,a,o)= \msg$ means that when agent $i$ performs
        action $a$ in state $s$ and obtains new sensor reading $o$,
        the information-exchange protocol broadcasts
        the message $\msg$ to the other agents; if $\msg= \bot$, then no  message is sent by $i$);

  \item $\delta_i: L_i \times A_i \times   \mathcal{P}(\cup_{j \in \Agents} M_j) \rightarrow \memory_i$
        is a function that updates the local memory as a function of the
        previous local state
        (comprised of the previous memory state and the previous sensor reading), an action, and a set of messages
        received.

\end{itemize}

An \emph{action protocol} $P$ for an information-exchange protocol $\exchange$,
is a tuple $\{P_{i}\}_{i\in\Agents}$ containing, for each agent $i$, $P_i : L_i
  \rightarrow A_i$ mapping the local states $L_i$ for agent $i$ in $\exchange_{i}$
to actions in $A_i$.

\section{Adversary model} \label{sec:adv}

Intersection protocols need to operate in an environment with several
forms of nondeterminism:
how messages are broadcast through the environment, failures of transmitters and receivers,
and the arrival pattern of vehicles. We model these aspects of the
environment in terms of an adversary.

The precise physics of the intersection may affect how broadcasts are transmitted through the environment.
Rather than attempt to model Euclidean distances and obstacles, we abstract the effects of these factors on
transmission. A \emph{transmission environment} is a relation $\canhear \subseteq (\Lanesin\times \Nat)^2$. Intuitively,
$((\ell,p), (\ell',p'))\in \canhear$ represents that, provided the agents' transmitters and receivers do not fail,
a message broadcast by an agent at position $p$ in lane $\ell$, will be received by an agent at position $p'$ in lane $\ell'$.
Transmission environments encode our assumption that the communication range is $\rho$.
We make one assumption about this relation: that for all
$\ell,\ell'\in \Lanesin$, we have $((\ell,0),(\ell',0))\in T$.
That is, messages broadcast by an agent at the front of
some lanes are received
(barring failure)  by all agents that are at the front of
any lane.

An \emph{adversary model} $\F$ is a set of adversaries;
formally, an \emph{adversary}
is a tuple
$\adversary = (\arrival, \canhear, F_t, F_r)$, where
$\arrival : \Agents
  \to \Nat \times \Lanesin \times \Lanesout$,
$\canhear$ is a transmission environment,
$F_t: \Nat\times \Agents \rightarrow \{0,1\}$, and
$F_r: \Nat \times \Agents \rightarrow \{0,1\}$.
Intuitively,
$\arrival$
is an \emph{arrival schedule}, which
describes when each
agent arrives in the system (i.e., enters a queue), its lane of
arrival, and its intended departure lane.
The function
$F_t$ represents failures of agents' transmitters and
the function
$F_r$
represents failures of agents' receivers.
$F_t(k,i) = 1$ means that if $i$ tries to broadcast in round
$k+1$ (i.e., between time $k$ and time $k+1$), then the broadcast will
be sent to all agents within range (i.e., within $\rho$ of $i$), and
perhaps others; similarly,
$F_r(k,j) = 1$ means that $j$ receives all broadcasts sent in round
$k+1$ by agents within range (but again, it may receive other
broadcasts as well).
Thus, a broadcast by agent $i$
in round $k+1$ is received by a $j$ within range of $i$ in round
$k+1$ iff $F_t(k,i) = F_r(k,j)
  = 1$.
The function $\tau$ describes when agents arrive in the system (which
we assume is under the control of the adversary).
In more detail, if
$\tau(j)= (k, (l_1,l_2))$, then at time $k$, agent
$j$ arrives in the system on lane $l_1$ with the intention of
departing on lane $l_2$. We assume that $\arrival$ is
\emph{conflict-free} in the sense that, for
all agents
$i \neq j$, if $\tau(i) =
  (k,(l_1,l_2))$ and $\tau(j) = (k,(l'_1,l'_2))$, then $l_1 \neq l'_1$. This
ensures that we do not have a conflict of two agents wanting to enter the same
queue for lane $l_1$ simultaneously.
(Exactly how this mutual exclusion of queue entry is assured is outside the scope of the model.
One way that it may come about is that vehicles approaching the intersection are already ordered along an
approaching lane.)

We consider adversary models that involve the following types of failures:
\begin{itemize}
  \item No failures ($\nofail$): the set of all adversaries
        $(\tau,\canhear,F_t,F_r)$
        where $F_r(k,i) = F_t(k,i) = 1$ for all $i \in
          \Agents$ and $k \in \Nat$.
  \item Crash failures ($\CR$): the set of all
        adversaries
        $(\tau, \canhear, F_t, F_r)$ where
        for all $i \in \Agents$ and $k \in
          \Nat$,
        (1)
        $F_t(k,i) = 0$ implies $F_t(k',i) = 0$ for all $k'>k$,
        and (2) $F_r(k,i) = 1$ for all $k$ and $i$.
  \item Sending omissions ($\SO$): the set of all
        adversaries
        $(\tau, \canhear, F_t, F_r)$ where
        for all $i \in \Agents$ and $k \in \Nat$,  $F_r(k,i)=1$.
\end{itemize}

An adversary model $\F$ has a \emph{fixed transmission
  environment} if all adversaries in $\F$
include the same transmission environment $\canhear$.
We believe that our techniques can be applied without change to the
general omissions case.

\section{Intersection Contexts} \label{sec:context}
A \emph{context} is a triple $(\exchange, \F, \pi)$ consisting of an
information-exchange protocol $\exchange$, an adversary model $\F$,
and an interpretation $\pi$.
To deal with intersections, we restrict
information-exchange protocols and interpretations so that they satisfy
certain conditions.
$(\exchange, \F, \pi)$ is an \emph{intersection context} if it
satisfies the following conditions:
\begin{itemize}
  \item
        The set of environment states $\envstates$ consists of states of the
        form $s_{e} = (\alpha, t, q_1, \dots, q_{|\Lanesin|}, \mathit{done})$
        where  $\alpha \in \F$ is an adversary, $t\in \Nat$ is a time,
        for each approach lane $l\in \Lanesin$,
        $q_l$ is a queue (list) of agents, intuitively the ones who have
        lane $i$ and not yet departed,
        and a set
        $\mathit{done} \subseteq \Agents$, representing
        the agents that have already passed through the
        intersection.
  \item
        The sensor model, in principle, could be defined to include
        information from a large variety of sensors
        and information sources, such as GPS,  in-road or road-side beacons,
        lidar, radar, or vision systems.
        We start with a minimal location-based sensor model, and leave it open
        for other fields to be added. Our minimal sensor model $\sensormodel = \{\sense_i\}_{i \in \Agents}$
        is defined so that  the sensor function $\sense_i$ maps environment
        states to tuples of the form
        $\langle \front_i,\lane_i , \plan_i \rangle $, where  $\front_i \in \{0,1\}$,
        $\lane_i  \in \Lanesin \cup \{\bot, \top\}$, and $\plan_i \in \Lanesout$.
        For $s_{e} = (\alpha, t, q_0, q_1, \dots, q_{|\Lanesin|}, \mathit{done})$,
        we have $\sense_i(s_e) = \langle \front_i,\lane_i, \plan_i \rangle$,
        where if $\tau$ is the arrival schedule in the adversary $\alpha$,
        \begin{itemize}
          \item $\pos_i$ maps from global states to $\Nat\cup\{\bot,\top\}$;
                $\pos_i(s_e) = \top$ if $i \in \mathit{done}$, $\pos_i(s_e) = k$ if
                there exists a queue $\ell$ such that
                $i$ is the $k$th position in queue $q_\ell$ (with the front of the
                queue counted as position 0),
                and $\pos_i(s_e) = \bot$ otherwise. (It follows from the state
                dynamics given below that
                $i$ is in at most one queue, so $\pos_i$ is well-defined.)
          \item $\front_i = 1$ iff $\pos_i(s_e) = 0$,
          \item if $i$ is in the queue $q_\ell$ for lane $\ell$,  then $\lane_i =
                  \ell$; if $i \in \mathit{done}$ then $\lane_i= \top$;
                and if $i \not \in \mathit{done}$ then $\lane_i = \bot$.
          \item
                if $\tau(i) = (k,(l,l'))$ then $\plan_i = l'$.
        \end{itemize}
        We have modelled an agent's intended departure lane
        $\plan_i$ as being received from
        the environment
        since, from the point of view of protocol design, this is part of the adversary.
  \item           The set of possible actions of agent $i$ in $\exchange_i$ is
        $A_{i} = \{\go, \noop\}$.
        Intuitively, $\go$
        represents that action of the agent making its planned move through the intersection. This action
        can be performed by agent $i$ only
        if $i$ is at the front of its queue.
        The action $\noop$ represents that the agent does not move,
        unless it is either scheduled for arrival in some queue, or
        in some position in a queue but not at the front, and the
        position before it is being vacated, in which case it
        advances in the queue.
  \item
        A global state is a tuple of the form
        $(s_{e},\{s_{i}\}_{i\in\Agents})$, where $s_e \in L_e$ and
        $s_i\in L_i$ for each agent $i\in \Agents$.
        An \emph{initial} global state  has
        \begin{itemize}
          \item
                $s_{e} = (\alpha, t, q_1, \dots, q_{|\Lanesin|}, \mathit{done})$,
                where $t=0$, each queue $q_l$ is empty, and
                $\mathit{done}$
                is the empty set, and
          \item for each agent $i\in \Agents$, the local state $s_i = (m_i,\sense_i(s_e))$ where $m_i \in  \initmemory_i$ is an initial memory state.
        \end{itemize}

  \item $\pi$ interprets the following
        atomic propositions
        based on the
        global state in the
        obvious way:
        $\front_i$, $\lane_{i} = l$ for $l \in \Lanesin$,
        $\plan_{i} = l$ for $l \in \Lanesout$,
        $\pos_{i} = k$ for $k \in \Nat \cup \{\bot,\top\}$.
\end{itemize}

Given an intersection  context $\gamma=(\exchange,\F,\pi)$ and a protocol $P$,
we construct an interpreted system $\I_{\gamma,P} = (\R_{\exchange,\F,P},\pi)$
representing all the possible behaviours of the protocol $P$ in context
$\gamma$.
The set $\R_{\exchange,\F,P}$ of runs
consists of all
runs $r$ that satisfy the following properties:
\begin{itemize}
  \item The initial state
        $r(0)$ of $r$ is an initial global state.
  \item
        For each $k\in \Nat$, the global state
        $r(k+1)=(s'_{e},\{s'_{i}\}_{i\in\Agents})$ is determined from
        $r(k)=(s_{e},\{s_{i}\}_{i\in\Agents})$
        by a procedure in which the order of events is as follows. First, the
        agents decide their actions (to go through the intersection or
        not). They then perform these actions, causing the
        queues to be updated; any newly arriving agents are also added to the
        queues in this step. The agents then take a sensor reading, from which
        they obtain new information about their position. This new information
        may be included in the message that an agent broadcasts.  Finally,
        each agent updates its memory state, based on their previous local
        state, the action performed, and the messages that were broadcast in
        the current round and received by the agent. We then proceed to the
        next round. Formally, state transitions are determined by the
        following procedure:
        \begin{itemize}
          \item First, each agent $i$ determines its action $P_i(s_i)$ according to the protocol $P$.
          \item
                If $s_e = (\alpha,m,q_1,\dots,q_{|\Lanesin|},\mathit{done})$, then we take $s'_e = (\alpha,  m+1, q'_1, \dots, q'_{|\Lanesin|}, \mathit{done}')$,  defined as follows. Note that the adversary $\alpha$ is the same in $s_e'$, and the time $m$ is incremented. Each queue $q'_\ell$ is obtained from $q_\ell$ by the following operations:
                \begin{itemize}
                  \item If $q_\ell(0) = i$ and $P_i(s_i) = \go$, then let $q''_\ell$ be the result of dequeueing agent $i$ from $q_\ell$.
                        Otherwise $q''_\ell = q_\ell$.
                  \item If $\arrival(i) = (m+1,(l_1,l_2))$ for any agent $i$, then we define
                        $q'_\ell = \mathit{enqueue}(i,q''_\ell)$, otherwise $q'_\ell= q''_\ell$.
                        (Recall that such an $i$ is unique, by assumption on $\arrival$.)
                \end{itemize}
          \item
                Finally, we take $\mathit{done}'$ to be the result of adding
                to the set  $\mathit{done}$ all agents $i$ who were at the front of
                any queue in $s_e$ such that $P_i(s_i) = \go$.
        \end{itemize}
  \item Next, for each agent $i$, we obtain a new sensor reading
        $\sense_i(s_e')$ of the updated state $s_e'$ of the environment.
        Using this sensor readings, each agent $i$ constructs the message
        $m_i = \mu_i(s_i,P_i(s_i),\sense_i(s_e'))$, which it
        broadcasts.
  \item For each agent $i$, we determine the set of messages $B^m_i$ that the agent
        receives in round $m+1$.
        If agent $i$ is not in any queue in state $s'_e$, or $F_r(m,i) = 0$ (agent $i$'s receiver fails in round $m+1$)
        then $B^m_i = \emptyset$.
        Otherwise, for each agent $i$ that is in a lane queue, let $\ell_i$ be the lane it is in and $p_i$ its
        position in the queue. We define
        $B^m_i $ to be the set of messages $m_j$ for which both
        $((p_j, \ell_j),(p_i,\ell_i)) \in T$
        ($j$'s transmission can be heard by agent $j$, given their positions) and $F_t(m,j) = 1$
        ($j$'s transmitter does not fail in this round.)

        \newcommand{\mem}{u} %
  \item Finally, if $s_i = (\mem_i,\sense_i(s_e))$, then $s'_i = (\mem'_i,\sense_i(s'_e))$,
        where
        $\mem'_i = \delta_i(s_i,P_i(s_i),B^m_i)$.
        (Note that we use the old sensor reading $\sense_i(s_e)$ to determine the new memory state, but
        not the new sensor reading $\sense_i(s'_e)$, since the latter will be
        visible to the agent in its new local state $s'_i$.)
\end{itemize}

$P$ is an \emph{intersection protocol}
for context $\gamma= (\exchange,\F,\pi)$
if the following are valid in $\I_{\gamma,P}$ for all $i,j \in
  \Agents$ where $i \neq j$, where
$\going_{i}$ is an abbreviation for $\front_i \land \bigcirc
  \neg \front_i$.
\begin{itemize}
  \item \textbf{Validity:}
        $\going_i \Rightarrow
          \front_i$.
  \item \textbf{Safety:} $(\going_{i} \land \going_{j}) \Rightarrow
          ((\lane_{i},\plan_{i}), (\lane_{j},\plan_{j})) \in \noconflict$.

  \item \textbf{Liveness:}
        $\front_{i} \Rightarrow \Diamond \going_{i}$.

\end{itemize}

Intuitively, \textbf{Validity} states that an agent does not
move through the intersection unless it is at the front
of
the queue in its lane. \textbf{Safety} states that if two agents go through the intersection at the same time, their moves are compatible and
do not cause a collision. (Note that the semantics of the action  $\go$ has been defined so as to ensure that an agent makes its planned
move, and not any other.) \textbf{Liveness} states that an agent
eventually gets to make its move through the intersection. (The model
implicitly
assumes that
vehicles do not have mechanical failures and block other vehicles in their lane.)

\section{Unnecessary waiting and optimality} \label{sec:unnecessary-waiting}
One desirable property of an intersection protocol is that it never makes
agents wait
unnecessarily. Eliminating unnecessary waiting is also a criterion
that has been considered
in the distributed mutual-exclusion literature
\cite{MosesP18}.
Intuitively, unnecessary waiting occurs if, given what happens in a
certain run
$r$, there is a point where if an agent had gone through the intersection
instead of waiting, safety would not be violated. In this section, we
define a notion of optimality that captures eliminating unnecessary
waiting.

We first give some definitions to define unnecessary waiting and
a domination-based notion of optimality.
For an intersection context $\gamma$ and protocol $P$,
\begin{itemize}
  \item $GO(r,m)$ is the set of agents that go through the intersection in round
        $m+1$,
        that is, the
        agents $i$ with $\I_{\gamma,P},(r,m) \models \going_i$.
  \item
        $\I_{\gamma,P},(r,m) \models \stg_i$ if $\I_{\gamma,P},(r,m) \models \pos_{i} =
          0$ and for all agents $j,k \in GO(r,m) \cup \{i\}$ where $j \ne k$,
        $(\lane_j(r,m),\plan_j(r,m))$ and $(\lane_k(r,m),\plan_k(r,m))$ are
        compatible
        moves according to $\noconflict$.

  \item For a run $r$ of a protocol $P$ in context $\gamma$, define $\gotime(r,i)$
        to be the time $m\in \Nat$ such that $\I_{\gamma,P},(r,m) \models \going_i$,
        and $\infty$ if there is no such time.
  \item $\front(r,m)$ is the set of agents that are in front of each queue,
        that is, the agents $i$ with $\front_i(r,m) = 1$.
\end{itemize}

\begin{definition}[unnecessary waiting]
  An intersection protocol
  $P$ has \emph{unnecessary waiting} with respect to an
  intersection context $\gamma$
  if there exists $i \in \Agents$ and point
  $(r,m)$ such that
  $\I_{\gamma,P},(r,m) \models \stg_i$ and $i\not\in GO(r,m)$.
\end{definition}

\begin{definition}[corresponding runs]
  Given action protocols $P,P'$ and context $\gamma$, two runs $r \in
    \I_{\gamma,P}$ and $r' \in \I_{\gamma,P'}$ \emph{correspond} if $r(0)
    = r'(0)$.
\end{definition}

Intuitively, corresponding runs have the same adversary, so agents
arrive at the intersection in the same sequence and at the same times
in the two runs.
We use this notion to define the following notion of one protocol
being better than another if it always ensures a faster flow of
traffic.

\begin{definition}[domination]
  An action protocol $P$ \emph{dominates} action protocol $P'$
  with respect
  to a context $\gamma$ if
  for all pairs of
  corresponding runs $r \in \I_{\gamma,P}$ and $r'
    \in \I_{\gamma,P'}$, all $i \in \Agents$,
  we have $\gotime(r,i) \leq \gotime(r',i)$.
  If $P$ dominates $P'$ but $P'$
  does not dominate $P$, then $P$ \emph{strictly dominates} $P'$.
\end{definition}

\begin{definition}[optimality]
  An intersection protocol $P$ is \emph{optimal} with respect to an
  intersection context $\gamma$ if there is no intersection protocol
  $P'$ that strictly dominates $P$ with respect to $\gamma$.
\end{definition}

Our goal is to connect the notions of unnecessary waiting and
optimality. The following result shows that the absence of
unnecessary waiting is sufficient for optimality.

\begin{propositionrep}
  \label{pro:unnecessary}
  If an intersection protocol $P$ has no unnecessary
  waiting with respect to an intersection context $\gamma$
  then $P$ is optimal with respect to $\gamma$.
\end{propositionrep}
\begin{proof}
  Suppose $P$ is not optimal. By definition, this means that there exists an
  intersection protocol $P'$ that strictly dominates $P$ with respect to
  $\gamma$. Then, by definition, there exist corresponding runs $r \in \I_{\gamma,P}$ and
  $r' \in \I_{\gamma,P'}$,
  an agent $i$, and a time $m$ such that $m=\gotime(r',i) < \gotime(r,i)$.
  Without loss of
  generality, suppose $m$ is minimal.
  Thus, for all $k<m$ and agents $j$, we have that if $k=
    \gotime(r',j)$ then $\gotime(r',j) \not < \gotime(r,j)$.
  However, since $P'$ dominates $P$, we have $\gotime(r',j) \leq \gotime(r,j)$,
  so we must in fact have $\gotime(r',j) = \gotime(r,j)$.
  Moreover, if $k= \gotime(r,j)$, then since $\gotime(r',j) \leq \gotime(r,j)$, we must also have $\gotime(r',j) = \gotime(r,j)$,
  by the minimality of $m$.  Thus, for all $k<m$, we have $GO(r',k) = GO(r,k)$.
  As these are corresponding runs,
  the arrival times and lanes are the same for all agents in $r$ and $r'$,
  as are the departure times up to time $m$.
  Thus, all agents are at exactly the same positions in $(r,m)$ as in $(r',m)$.
  In particular, $\front_i(r,m) = 1$.
  Since $P'$
  dominates $P$, $GO(r,m)
    \subseteq GO(r',m)$,
  so  $GO(r,m) \cup\{i\} \subseteq GO(r',m)$, by choice of $i$ and $m$.
  Since $P'$ is safe, for all agents $j,k \in GO(r,m) \cup \{i\}$
  such that $j \ne k$,
  $(\lane_j(r,m),\plan_j(r,m))$ and $(\lane_k(r,m),\plan_k(r,m))$ are compatible
  moves according to $\noconflict$.
  It follows that $\I_{\gamma,P},(r,m) \models \stg_i$.
  Because $i \not \in GO(r,m)$,
  there is unnecessary waiting in $P$.
\end{proof}

{F}rom here on, we consider contexts that require some conditions on
broadcasting.
This is because if not enough information is exchanged or
adversaries are too powerful, we cannot have a protocol that avoids
unnecessary waiting.
To see why, consider a setting where the intersection has two incoming
lanes
and one
outgoing lane, each agent has access to a global clock, and
the information-exchange protocol does not send any messages.
While a correct protocol exists that uses the global clock to determine
when an agent at the front of a queue can proceed to the intersection
(essentially, we use the global clock to simulate a traffic light, and
have the agents proceed in turns),
unnecessary waiting cannot be eliminated, simply
because the agents do not exchange enough information to rule out
safety violations.

However, even with full information exchange where each agent broadcasts
its entire local state in each round and records every broadcast it receives,
the converse of Proposition~\ref{pro:unnecessary} still does not hold.
A protocol
may have unnecessary waiting and still be optimal even with
full information exchange.

\begin{propositionrep} \label{pro:counterexample}
  There exists an intersection context $\gamma$ with
  full information exchange and no failures and an intersection protocol
  $P$ such that
  $P$ has unnecessary waiting and is optimal
  with respect to $\gamma$.
\end{propositionrep}
\begin{proof}
  Consider the intersection context $\gamma$ where the agents use
  full-information exchange, there are no failures,
  the   information-exchange protocol is such that
  agents have access to a global clock, and the time is part of their
  state, and the transmission
  environment is such that
  agents who are not in the front of some queue do not receive
  messages from agents who are leaving the intersection.
  Consider an intersection protocol $P$ with
  unnecessary waiting only in a run where a particular agent $j$ is at
  the front of its lane at time 1, and there are no other cars waiting.
  Rather than going in round 2, $j$ goes in round 4; it waits
  unnecessarily for two rounds.  Moreover, there
  is a run $r_1$ where, in addition, in round 3, vehicles $j_1$ and
  $j_2$ arrive in different lanes and go directly to the front of their
  respective lanes; moreover, their intended departure lanes
  conflict with each other, but
  not
  with that of $j$.  No vehicles other than $j$, $j_1$, and $j_2$ arrive
  before round 3.  There is also a run $r_2$
  where vehicles $j_1$ and $j_2$ again arrive at the
  intersection in round $3$, but no other vehicles have arrived up to then.
  (We describe only as much of $P$ as is needed to illustrate the point.)
  Suppose that $P$
  behaves differently if (1) $j_1$ and $j_2$ come to the intersection,
  and there is noone else present at time $3$, and
  (2) $j_1$ and $j_2$ come to the intersection,
  and they are aware of another agent $j$'s presence at time $2$,
  but their intended departure lanes do not conflict with that of $j$.
  Specifically, in the
  first case,
  with $P$, $j_2$ goes first; in the second, $j_1$ goes first.
  Thus, with $P$, $j_1$ goes in round
  $4$ of $r_1$,  and $j_2$ goes in round $4$ of $r_2$.
  This scenario is illustrated in Figure~\ref{fig:example}.
  \begin{figure}[htp]
    \centering
    \includegraphics[scale=0.3]{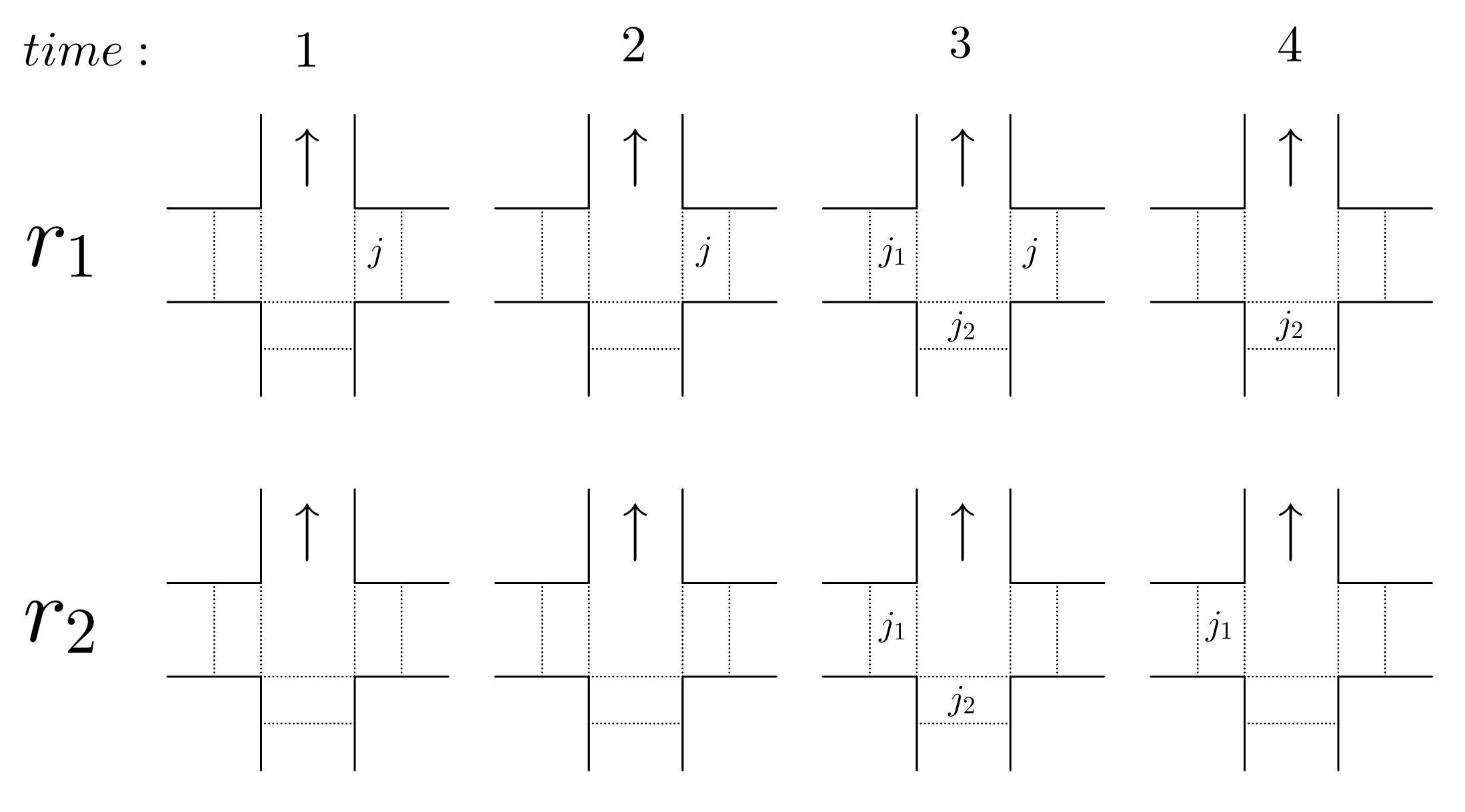}
    \caption{Unnecessary waiting and domination.
      Illustration of the scenario in the proof of
      Proposition~\ref{pro:counterexample}.}
    \label{fig:example}
  \end{figure}

  Now suppose, by way of contradiction, that protocol $P'$ dominates $P$.
  Then it must have no unnecessary waiting in the scenario above, since
  this is the only scenario where $P$ has unnecessary waiting.
  (Otherwise, the argument of Proposition~\ref{pro:unnecessary} shows
  that $P'$ does not dominate $P$.)  If $r_1'$ is the run of $P'$
  corresponding to $r_1$, $j$ must go at round 2 of $r_1'$.  By our
  assumptions on $\gamma$, when $j_1$ and $j_2$ reach the front of their
  lanes in round 2 of $r_1'$,
  they are not aware that $j$ was
  there in the previous round.
  That means that $j_1$ and $j_2$ cannot distinguish
  $(r_1',2)$ from $(r_2',2)$, where $r_2'$ is the run of $P'$ corresponding
  to $r_2$.
  Thus, they
  must do the same thing in both runs, whereas they do different things with
  $P$ in round 4.   So
  $P'$ does not dominate $P$.  It follows that no protocol dominates
  $P$, so $P$ is optimal.
\end{proof}

Proposition~\ref{pro:counterexample} suggests that the definition of
optimality doesn't exactly
capture the
lack of unnecessary waiting. We thus consider another definition
that we call \emph{lexicographic optimality}.

\begin{definition}[lexicographical domination]
  An action protocol $P$ \emph{lexicographically dominates}
  action protocol
  $P'$ with respect to a context $\gamma$ if for
  all
  corresponding runs $r \in
    \I_{\gamma,P}$ and $r' \in \I_{\gamma,P'}$,
  either $GO(r,m) = GO(r',m)$ for all times $m$ or,
  at the first time $m$ when
  $GO(r,m) \neq GO(r',m)$, we have
  $GO(r',m) \subsetneq GO(r,m)$.
  If $P$
  lexicographically dominates $P'$ but $P'$ does not lexicographically dominate
  $P$, then $P$ \emph{strictly}
  lexicographically dominates $P'$.
\end{definition}

\begin{definition}[lexicographic optimality]
  An intersection protocol $P$ is \emph{lexicographically
    optimal} with respect
  to an intersection context $\gamma$ if there is no intersection
  protocol $P'$ that strictly
  lexicographically dominates $P$ with respect
  to $\gamma$.
\end{definition}

\begin{propositionrep} \label{prop:uw-implies-lo}
  If an intersection protocol $P$ has no unnecessary waiting with respect to
  an intersection context $\gamma$,
  then $P$ is lexicographically optimal with respect to $\gamma$.
\end{propositionrep}
\begin{proof}
  We prove the contrapositive. Suppose that $P$ is an intersection protocol
  that is not lexicographically
  optimal. Then there exists an
  intersection protocol $P'$ that strictly
  lexicographically
  dominates $P$. It is immediate from the definition of lexicographic
  domination that there exist corresponding runs
  $r$ and $r'$ and a time $m$
  such that $GO(r,m) \subsetneq
    GO(r',m)$
  and $GO(r,k) = GO(r',k)$ for $k<m$. It follows from the latter condition and the fact that arrival times and lanes are the same for all
  agents in these runs that all agents are in the same position in $(r,m)$ and $(r',m)$.
  Hence, at least one agent
  has unnecessary waiting.
\end{proof}

The following result provides a partial converse to
Proposition~\ref{prop:uw-implies-lo}.

\begin{propositionrep} \label{prop:lo-implies-uw}
  If an intersection protocol $P$ is lexicographically optimal with respect to
  an information context $\gamma$ with full information exchange and no
  failures,
  then $P$ has no unnecessary waiting with respect to $\gamma$.
\end{propositionrep}
\begin{proof}
  We prove the contrapositive. Suppose $P$ is an intersection protocol that
  has unnecessary waiting with respect to
  $\gamma$ with full information exchange and no failures.
  We show that there exists a protocol $P'$ that strictly lexicographically dominates $P$.
  There exists a run $r$ of $P$
  and time $m$ where an agent $j$ waits
  unnecessarily. Let $\omega$ be the local state of agent $j$ at
  $(r,m)$.
  The sensor reading component of $\omega$
  includes     $\lane_i =   l$,
  $\front_i$
  and $\plan_i = l'$ for some lanes $l$ and $l'$.

  Consider the following protocol $P'$,
  where $P^T$ is a
  ``traffic light''
  protocol defined by the rule that agents with
  $\front_i$ in lane
  $(m \mod |\Lanesin|)$
  perform the action $\go$, and
  $\psi_i$ is defined to be true if agent $i$ hears
  either directly from $j$ or from its successor in the
  queue that $j$ was in local state $\omega$.
  (Note that if the domain of the variable $\Time_i$ in agent $i$'s
  memory state is $\Nat$, then the
  agent always knows the time $m$.
  If the domain of this variable is $\mathbb{Z}_n$ then
  we assume here that $n$ is a multiple of $|\Lanesin|$, so that the agent can determine
  $(m \mod |\Lanesin|)$.)
  \\ \begin{program}[H]
    \DontPrintSemicolon
    \lIf{$i = j \land
        \mathit{local\_state}_i = \omega$}{$\go$}
    \lElseIf{$\psi_i$}{$P^T$}
    \lElse{$P$}
    \caption{$P'_i$}
  \end{program}

  Clearly, in runs $r$ where agent $j$
  is never in local state $\omega$, $P'$ is equivalent to $P$ (since
  $\psi_i$ never holds in such runs $r$). In runs where
  $j$ is in local state $\omega$ at some time
  $m'$
  there is only one such time since only agent $j$'s actions are modified for $P'$.
  The set of agents that
  go through the intersection is identical before time $m'$, and $P'$ sends a
  strict superset through the intersection at time $m'$.
  In particular, the run corresponding to $r$ for $P'$ is such a run. This shows that $P'$ satisfies the
  conditions for strictly lexicographically dominating $P$.

  However, to conclude that $P$ is not lexicographically optimal, we also need
  to show that $P'$ is an intersection protocol. This is trivial for runs on
  which $P'$ is equivalent to $P$, since $P$ is an intersection protocol. On
  runs where $j$ is in local state $\omega$ at time $m$, we argue as follows.
  Note that for all times up to time $m$ in such runs, all agents run
  $P$, which satisfies validity and safety.
  Round $m+1$ also preserves these properties, since agents other than
  $j$ run $P$, so do not violate validity or safety.
  Agent $j$ also goes in this round, but it is at position 0 in its queue and
  its move is compatible with that of the other agents that go, so
  this agent also does not violate validity or safety.

  To prove validity and safety for later rounds, we prove by induction
  on $m'$ that for all times $m' >m$, either
  all agents $i$ at the front of some lane at time $m'$ satisfy
  $\psi_i$ and are running $P^T$, or no agent $i$ in any lane
  at time $m'$
  satisfies $\psi_i$, all agents are running $P$, and at some  time
  $m''$ with $m' \ge m'' > m$,
  there are no agents in any queue at time $m''$.  For the base case,
  $m' = m+1$.  By our assumptions, $j$ broadcasts as it goes into the
  intersection, and all agents $i$ at the front of some queue receive $j$'s
  message (since $j$ does not fail) and thus satisfy $\psi_i$.  Now
  consider an arbitrary time $m' > m+1$.  If, at time $m'-1$,
  $\psi_i$ holds for some agent $i$ at the front of some queue then, by
  the induction hypothesis, all the agents at the front of some queue
  are running $P^T$.  Whichever agents move up to the
  front of the queue in round $m'$ will hear about $j$ (from the
  agents that were at the front of the queue at time $m'+1$), and thus
  the claim holds at the end of round $m1$ (i.e., at time $m'$).  On
  the other hand, if
  there is no agent $i$ for which $\psi_i$ holds at time $m'-1$, this will
  clearly also be the case at time $m'$.  Moreover, it follows from
  the induction assumption that there is some time $m''$ with $m'-1
    \ge m'' > m$ where there were no agents $i$ in any queue. Again,
  using the induction hypothesis, all the agents in any queue at time
  $m'$ are clearly running $P$.  (Since $P$ is the default protocol,
  all agents that join some queue in round $m'$ also run $P$ at time $m'$.)

  Since both $P$ and $P^T$ satisfy safety and validity, safety and
  validity continue to hold for all times $m' > m$.  There is one minor
  subtle point that needs to be taken care of:  If all agents are
  running $P$, we have to show that the global state is one that
  could actually have arisen with protocol $P$, since otherwise the fact
  that $P$ is an intersection protocol on its own states tells us
  nothing!  This follows from the fact that if all agents are running
  $P$ at time $m'$, there must have been some earlier time $m''$ when
  there were no agents in any queue.  Thus, the global state at time $m'$
  is one that could have arisen from using $P$ starting at time 0.

  The analysis above shows that with $P'$, we either have $P^T$ run for the rest
  of the run after $m$, or agents run $P^T$ for a while and then switch to running
  $P$ for the rest of the run. In either case, since $P^T$ and $P$ both satisfy
  liveness, we also get liveness for $P'$.
\end{proof}

While considering a full-information context shows that lexicographic
optimality captures the condition on unnecessary waiting better,
it is also possible to get a similar result in a context with much less
information exchange, even without a global clock.

We say that an intersection context $\gamma = (\exchange,\F,\pi)$ is
\emph{sufficiently rich} if $\exchange$ satisfies the following conditions:
\begin{itemize}
  \item In round $m$, if agent $i$ is going to be at the front of some
        lane at
        time $m$,
        then $i$ broadcasts a message encoding $\lane_i$, $\plan_i$.
        (Note that we are here using the fact that in agent's message in
        round $m$ can incorporate the effect of its round $m$ action.
        Thus, if an agent $i$ moves to the front of the queue for some lane
        in round $m$, then $i$ will sense that it is at the front of the
        queue, and $i$ can send a message in round $m$ saying that it is
        about to be at the head of the queue for its lane.)
  \item Each agent records the $(\lane,\plan)$ pair for each agent
        in the front  of a queue, and either no agents in the
        queue other than those at the front broadcast, or agents
        at the front of a queue tag their messages to indicate
        that they are at the front of their queue.
\end{itemize}

Intuitively, if an intersection context is sufficiently rich, in the round
$m$ that an agent $i$ reaches the front of the queue for some lane, it
knows about all other agents that are in the front of their queues at
time $m$, and knows their intentions (if there are no failures).
\begin{lemmarep} \label{lem:knows-lanes}
  If $\gamma$ is a sufficiently rich intersection context with no
  failures, $P$ is an
  intersection protocol, and
  $\front_i(r,m) = 1$, then
  \begin{align*}
     & \I_{\gamma,P} \models
    \forall l \in \Lanesin(
    K_{i}(\exists j \in \Agents\; \exists l' \in \Lanes \; (\front_{j}
    \land \lane_{j} = l \land \plan_j = l') \; \lor \\
     & \qquad \qquad
    \qquad \quad \;
    K_i(\forall j \in \Agents \; (\lane_{j} \neq l))
    )
    ).
  \end{align*}
\end{lemmarep}
\begin{proof}
  Let $r$ be a run in $\I_{\gamma,P}$ where $\gamma$ is a sufficiently
  rich intersection context with no failures.
  We proceed by induction on $m$. In the base case, since
  $\gamma$ is
  an intersection context,
  we know that all queues are empty and $\I_{\gamma,P},(r,0) \models
    \pos_{i} = \bot$
  for all $i \in \Agents$.
  Therefore, $\front(r,0) = \emptyset$ and the claim holds vacuously.

  Now suppose that the claim holds for time $m-1$ of $r$.
  If $\front_i(r,m) = 1$, because $\gamma$ is sufficiently rich, $i$
  hears in round $m$ from all of
  the agents in
  $\front(r,m)$ about their current lane and intent. Moreover,
  if $i$ hears from other agents in the queue for some lane, $i$
  will hear which of the agents it hears from are in the front of their
  queues.
  Thus, $i$ knows exactly which lanes have an agent in front and
  what their intent is.
\end{proof}

Given a sufficiently rich intersection context $\gamma$, all protocols
that we
care about will depend only on what the agents hear from agents at the front
of each queue.
We say that an intersection protocol $P$ \emph{depends only on agents in the
  front of their queues} in intersection context $\gamma =
  (\exchange,\nofail,\pi)$
if,  for all $i \in \Agents$, the following condition holds: for all
pairs $s_i,s'_i$ of
possible local states of agent $i$ drawn from $L_i$ in $\exchange$, if
$\mathit{front}(r,m) = \mathit{front}(r,m')$, then
$P_i(r_i(m)) = P_i(r'_i(m'))$. Note that this condition makes sense
only in
a sufficiently rich intersection context in the no-failures setting,
since otherwise an agent may not know which agents are at the front of
their queues, so its protocol cannot depend on this fact.

\begin{propositionrep} \label{prop:lo-implies-uw-nonfip}
  Let $\gamma$ be a sufficiently rich intersection context with no failures.
  If an intersection protocol $P$ is lexicographically optimal with respect to
  $\gamma$ and $P$ depends only on agents in the front of
  their queues,
  then $P$ has no unnecessary waiting with respect to $\gamma$.
\end{propositionrep}
\begin{proof}
  We prove the contrapositive and proceed similar to the proof of
  Proposition~\ref{prop:lo-implies-uw}. Suppose intersection protocol $P$
  depends only on agents in the front of their queues and has unnecessary
  waiting with respect to sufficiently rich intersection context $\gamma$ with
  no failures. We show that there exists a protocol $P'$ that strictly
  lexicographically dominates $P$ with respect to $\gamma$. There exists a run
  $r$ of $P$ and time $m$ where an agent $j$ waits unnecessarily. Let $\omega$
  be the local state of agent $j$ at $(r,m)$.
  We then consider the following protocol $P'$: \\
  \begin{program}[H]
    \DontPrintSemicolon
    \lIf{$i = j \land \mathit{local\_state}_i = \omega$}{$\go$}
    \lElse{$P$}
    \caption{$P'_i$}
  \end{program}

  We first argue that $P'$ lexicographically dominates $P$ with respect to
  $\gamma$. In runs where agent $j$ is never in local state $\omega$, both $P$
  and $P'$ are equivalent. Hence, consider a run $r$ in $\I_{\gamma,P}$ where
  agent $j$ satisfies $r_j(m) = \omega$ for some time $m$ and let $r'$ be the
  corresponding run in $\I_{\gamma,P'}$. For $m'$ such that $m' < m$, we
  directly have $GO(r,m') = GO(r',m')$ since $P$ and $P'$ are equivalent
  before $m$. At $m$, we have $GO(r,m) \subsetneq GO(r',m)$ since $j$ goes
  with $P'$ and not with $P$.

  It remains to show that $P'$ is an intersection protocol.
  Since $P$ is an intersection protocol, validity, safety, and liveness
  hold for all runs of $P'$ where $j$ never attains the local state $\omega$.
  Now, consider $r' \in \I_{\gamma,P'}$ where $m^*$ is the first time
  when $r'_j(m^*) = \omega$. Validity and safety also hold for all
  points $(r',m'')$ such that $m'' \leq m^*$, since $P$ and $P'$ are
  equivalent
  before $m^*$. For $m'' > m^*$, observe that there must exist a
  run $r_f \in \I_{\gamma,P}$ where $\front(r_f,m_f) = \front(r',m'')$.
  This is because for any given set of agents in the front, there exists a
  run where exactly that set of agents arrive at the same time to an empty
  intersection with the same preferences. The assumption that $P$
  depends only on agents in the front then guarantees that $P'$ at $(r',m'')$
  is equivalent to $P$ at $(r_f,m_f)$. Hence, for all $m'' > m^*$,
  $P'$ satisfies validity and safety, since $P$ does so for all points.
  Lastly, we observe that $r_f$ can be constructed such that after
  $(r_f,m_f)$, the agents that arrive exactly match $(r',m'')$.
  Therefore and inductive argument shows that the liveness of $P$ implies
  the liveness of $P'$ in $r'$.
\end{proof}

\section{Intersection policies} \label{sec:intersection-policies}

Intuitively, an  \emph{intersection policy} describes which moves are permitted, as a function of
a history describing
what happened in the run until that point in time (in particular, the nondeterministic
choices that have been made by the adversary up to that moment of time),
but excluding details of the agent's local states and protocol.

We will use intersection policies as a tool to design standard
protocols that solve the intersection problem. Roughly, the methodology is
the following. Initially, we will design an intersection policy
$\sigma$ that guarantees safety and liveness for agents complying with $\sigma$.
We will then find standard intersection protocols that implement a
knowledge-based program using $\sigma$.
Finally, we will show that every intersection protocol can be
obtained in this way.

A history captures the nondeterministic choices made by the adversary up to some moment of time.
Given an adversary
$\alpha = (\tau,T,F_t,F_r)$
for a context $\gamma$ and natural number
$m\in \Nat$, define  the \emph{choices of $\alpha$
  in round $m+1$} to be the
tuple
$\alpha_m = (\tau^m,T,F_r^m,F_t^m)$, where
$\tau^m =\{(i,\ell,\ell') \in \Agents\times \Lanesin\times\Lanesout~|~ \tau(i) = (m+1,\ell,\ell')\}$,
and for $a=r$ and $a=t$,
the function $F_a^m: \Agents \rightarrow \{0,1\}$ is
defined by $F_a^m(i) = F_a(m,i)$.
(Recall that the transmission environment $T$ is fixed for the run, so
the same $T$ applies in each round.)
An  \emph{adversary history} is a finite sequence of such tuples; for an
adversary $\alpha$ and time $m$,
define $H(\alpha,m) = \langle \alpha_0, \ldots, \alpha_{m-1} \rangle$.
(If $m=0$, $H(\alpha,m)$ is the empty sequence.)
Given a context $\gamma$,
$\histories_\gamma$ is the set of all adversary histories
$H(\alpha,m)$ such that
$\alpha$ is an adversary for $\gamma$ and $m\geq 0$.
If $r$ is a run of context $\gamma$ with adversary $\alpha$, we also
write $H(r,m)$ for $H(\alpha,m)$.

\begin{definition}[intersection policy]
  An \emph{intersection policy} for a context $\gamma$ is a mapping $\sigma : \histories_\gamma \to
    \mathcal{P}(\Lanesin \times \Lanesout)$.
\end{definition}

Intuitively, an intersection policy says which moves are
\emph{permitted} in the given round.
An agent at the front of a queue for lane $\ell$ \emph{may} go if its
intent
is to make move to lane
$\ell'$ and the move $(\ell,\ell')$ is permitted.  (However,
in contexts with failures,
the agent may fail to go because it does not know that  its move is permitted.)

An infinite sequence $h_0, h_1, \ldots $ is \emph{feasible} in a context $\gamma$ if there exists an
adversary $\alpha$ of $\gamma$ such that $h_m = H(\alpha,m)$ for all $m\geq 0$.
An intersection policy $\sigma$ for a context $\gamma$ is \emph{correct}
for a context $\gamma$
if it satisfies the following
specification:
\begin{itemize}
  \item \textbf{Conflict-free:}  For all histories $h\in \histories_\gamma$, and agents $i\neq j$, if $(l_{i},l_{i}'), (l_{j},l_{j}')
          \in \sigma(h)$ then $(l_{i},l_{i}',l_{j},l_{j}') \in \noconflict$.

  \item \textbf{Fairness:} For all feasible infinite sequences of histories
        $h_0 , h_1, h_2 \ldots$, all moves
        $(\ell,\ell') \in \Lanesin\times\Lanesout$, and all $m\geq 0$,
        there exists $m'\geq m$ such that   $(\ell, \ell') \in \sigma(h_{m'})$.
\end{itemize}

Intuitively, an intersection policy $\sigma$ is conflict-free if
$\sigma$ never permits a
conflicting set of moves to occur simultaneously.
An intersection policy $\sigma$ is fair if, in every feasible
infinite sequence of histories, $\sigma$
permits every possible move infinitely often.
A context $\gamma$ is \emph{$\sigma$-aware} for an intersection
policy $\sigma$ if,
for all protocols $P$ for $\gamma$, agents $i$, lanes $\ell \in \Lanesin$ and $\ell'\in \Lanesout$, we have
$\I_{\gamma,P} \models ((\ell,\ell')\in \sigma \land \lane_i = \ell) \rimp
  K_i((\ell,\ell')\in \sigma)$.

\begin{example}
  A simple correct intersection policy is a cyclic traffic
  light. Suppose that the set of all moves $\Lanesin\times \Lanesout$
  is partitioned into a collection $S_0,\ldots ,S_{K-1}$, such that each set $S_k$ is a compatible set of moves.
  Then the intersection policy defined on histories $h$ by $\sigma(h) =
    S_{|h| \mod K}$ is easily seen to be correct
  (whatever the context $\gamma$).
  Clearly, every synchronous context is $\sigma$-aware for this policy.
\end{example}
\begin{example}
  A more complicated intersection policy is one that prioritizes
  certain lanes if they contain specific agents (e.g., an
  ambulance).
  Suppose that $\mathcal{A} \subseteq \Agents$ is a finite set of
  higher-priority agents.  Consider the intersection policy that allows
  moves given by a cyclic
  traffic-light policy unless there is an agent in $\mathcal{A}$
  that has arrived and is yet to make a move. In that case,
  the policy runs the traffic-light policy restricted to lanes containing
  higher-priority agents.    This requires
  considering past moves permitted by the policy and
  the adversary history to determine the state of the queues.
  In a context with no failures, synchrony, and
  a transmission environment such that the
  presence of a
  higher-priority agent is known by agents in the front, we get
  $\sigma$-awareness.
\end{example}

Given an intersection policy $\sigma$,
consider the following knowledge-based program $\kbp^{\sigma}$: \\
\begin{program}[H]
  \DontPrintSemicolon
  \lIf{$K_{i}(\front_i \land (\lane_{i},\plan_{i}) \in \sigma)$}{$\go$}
  \lElse{$\noop$}
  \caption{$\kbp^{\sigma}_i$}
\end{program}

Here the formula $(\lane_{i},\plan_{i}) \in \sigma$ is satisfied at a point $(r,m)$
if we have $(\lane_{i}(r,m),\allowbreak\plan_{i}(r,m))
  \in \sigma(H(r,m))$.

An action protocol $P$ \emph{implements} a knowledge-based program
of the form ``if $K_i\phi$ then $\go$ else $\noop$'' in a context
$\gamma$ if, for all points $(r,m)$ of $\I_{\gamma,P}$, we have
$P_i(r_i(m)) = \go$ iff $\I_{\gamma,P}(r,m) \models K_i \phi$. (See
\cite{FHMV} for the definition for more general program structures.)

We immediately get the
following.

\begin{propositionrep} \label{prop:global-safe}
  For every synchronous context $\gamma$ and intersection policy $\sigma$ for $\gamma$, there
  exists a behaviorally unique\footnote{``Behavioral uniqueness'' here means that any two implementations take the same actions at all reachable states, and can differ only on unreachable states.}
  $P$ implementing  the knowledge-based program $\kbp^\sigma$ with
  respect to $\gamma$.
  If $\sigma$ is a correct intersection policy with respect to
  $\gamma$, then
  every implementation $P$  of the knowledge-based program
  $\kbp^{\sigma}$ with respect to $\gamma$
  satisfies safety and validity.
\end{propositionrep}
\begin{proof}
  The formula $\phi = \front_i \land (\lane_{i},\plan_{i}) \in \sigma$ in the knowledge condition of $\kbp^\sigma$
  depends only upon the past and the current state of the system.
  Existence and uniqueness of implementations in synchronous contexts $\gamma$ therefore follows from
  Theorem~7.2.4 of \cite{FHMV}. This theorem
  implies that in synchronous contexts, a knowledge-based program, containing only knowledge conditions with
  this property, has an implementation, and all its implementations are behaviourally equivalent
  (This means that if $P$ and $P'$ are implementations, then $P_i(s) = P'_i(s)$ for all local states $s$
  that are reachable in $\I_{\gamma,P}$, and in fact, $\I_{\gamma,P} = \I_{\gamma,P'}$.)

  Suppose that $\sigma$ is correct with respect to an intersection context $\gamma$ and $P$ implements
  $\kbp^{\sigma}$  with respect to $\gamma$.
  If, for agents $i\neq i'$,
  we have $\I_{\gamma,P},(r,m) \models \going_{i} \land \going_{i'}$, then,
  since $\gamma$ is an intersection context   it follows that  $P_{j}(r_{j}(m))=$ $\go$   for $j \in \{i,i'\}$.
  Since $P$ implements $\kbp^{\sigma}$,  we must have
  $\I_{\gamma,P},(r,m) \models K_{j}(
    \front_i
    \land  (\lane_{j}, \plan_{j}) \in \sigma)$.
  Since
  $K_j\phi \rimp \phi$
  is valid, it follows that $\I_{\gamma,P},(r,m) \models
    \front_j \land
    (\lane_{j}, \plan_{j}) \in \sigma$.
  Since $\sigma$ is conflict-free, we have that
  $$((\lane_i(r,m), \plan_i(r,m)), (\lane_{i'}(r,m), \plan_{i'}(r,m))
    \in \noconflict.$$
  Thus, $P$ is safe.

  Now suppose that for agent $i$, we have  $\I_{\gamma,P},(r,m) \models \going_i$.
  Then we have $P_{i}(r_{i}(m))  = \go$, from which it follows, as in the
  argument above, that
  $\I_{\gamma,P},(r,m) \models
    \front_j$.
  Thus, $P$ satisfies validity.
\end{proof}

Proposition~\ref{prop:global-safe} provides a way of deriving an
intersection protocol from an intersection policy.
We can also show that every intersection protocol can be
derived from some intersection policy in this way.

\begin{propositionrep} \label{prop:global-exists}
  If $P$ is
  a protocol satisfying validity and safety
  then there exists
  a conflict-free intersection policy $\sigma$
  for $\gamma$
  such that $P$ implements $\kbp^{\sigma}$ with respect to $\gamma$.
\end{propositionrep}
\begin{proof}
  Suppose that $P$ is an intersection protocol  for an intersection context
  $\gamma$.  We show how to construct the intersection policy $\sigma$.

  Note first that for two adversaries $\alpha$ and $\alpha'$ of $\gamma$ with
  $H(\alpha,m) = H(\alpha',m)$,
  if $r$ and $r'$ are runs of $\I_{P,\gamma}$ with adversaries $\alpha$
  and $\alpha'$, respectively,
  then  $r(m) = r'(m)$. The proof is a straightforward induction,
  using the fact that $P$ is deterministic and that
  the state transition semantics is a deterministic function of the actions selected by $P$ and
  the failures and arrival times selected by the adversary.
  We can therefore define the function $\sigma$ with domain
  $\histories_\gamma$ by
  taking
  $\sigma(h) = \sigma'(r,|h|)$ where $r$ is a run of $\I_{P,\gamma}$
  with $H(r,|h|) = h$ and
  $$
    \sigma'(r,m) = \{ (\ell,\ell')~|~ \exists i\in \Agents
    \left(\begin{array}[h]{l}
        \front_i(r,m) = 1
        \text,\,  P_i(r_i(m)) = \go, \, \\
        \lane_i(r,m) = \ell, \, \plan_i(r,m) = \ell'%
      \end{array}\right)
    \}.$$
  To see that $\sigma$ is well-defined, note that by determinism of $P$, and the
  fact that the history $h=H(r,m)$ captures all the non-determinism to time $m=|h|$,
  for all points $(r,m)$ and $(r',m)$ of $P$ with $H(r,m) = H(r',m)$,
  we have $r[0\ldots m] = r'[0\ldots m]$, and consequently $\sigma'(r,m) = \sigma'(r',m)$.
  It is obvious from the fact that $P$ is safe that $\sigma$ is conflict-free.

  To show that $P$ implements $\kbp^\sigma$, we need to show that for all points $(r,m)$ of
  $\I_{\gamma,P}$ and agents $i$, we have that $P_i(r_i(m)) =
    (\kbp^{\sigma}_{i})^{\I_{\gamma,P}}(r_i(m))$.
  We consider two cases, depending on whether
  $(\kbp^\sigma_i)^{\I_{\gamma,P}}(r_i(m)) = \noop$.

  If $(\kbp^\sigma_i)^{\I_{\gamma,P}}(r_i(m)) = a \neq
    \noop$, then $a = \go$ and  $\I_{\gamma,P},(r,m) \models
    K_i( \front_i \land (\lane_i,\allowbreak\plan_i)\in \sigma )$.
  Hence,  $\I_{\gamma,P},(r,m) \models \front_i \land  (\lane_i,\plan_i)\in \sigma$,
  so, by the semantics of this formula,
  we have
  $\front_i(r,m) = 1$ and
  $(\lane_i(r,m),\plan_i(r,m))) \in
    \sigma'(r,m)$.
  By the definition of $\sigma$ and the fact that $r$ is a run with
  adversary $\alpha$,
  it follows that there exists an agent $j$ such that
  $\front_j(r,m) = 1$,
  $P_j(r_j(m)) = \go$,
  $\lane_j(r,m) = \lane_i(r,m)$, and $\plan_j(r,m) =  \plan_i(r,m)$.
  Since there can only be at most one agent
  at the front of
  a given lane,
  we must have that $j=i$. Hence
  $P_i(r_i(m)) = \go = \kbp_{\sched,i}^{\I_{\gamma,P}}(r_i(m))$.

  If $\kbp_{\sched,i}^{\I_{\gamma,P}}(r_i(m)) = \noop$, then
  $\I_{\gamma,P},(r,m) \models \neg K_i( \front_i \land \allowbreak (\lane_i,\plan_i)\in \sigma ) $.
  The proposition $\front_i$ is local to agent $i$, since $\front_i$ is included in the sensor
  reading part of the agent's local state. Thus $\front_i \dimp K_i(\front_i)$ is valid.
  If
  $\front_i(r,m) = 0$,
  we must have that $P_i(r_i(m)) = \noop$, since $P$ is valid.
  It therefore suffices to show that if
  $\I_{\gamma,P},(r,m) \models \front_i  \land \neg K_i( (\lane_i,\plan_i)\in \sigma ) $
  then $P_i(r_i(m)) = \noop$.
  We do so by showing the contrapositive, specifically, that
  if  $\front_i(r,m) =1$ and $P_i(r_i(m)) \neq \noop$ then
  $\I_{\gamma,P},(r,m)\models K_i( (\lane_i,\plan_i)\in \sigma ) $.
  Note that if  $P_i(r_i(m)) \neq \noop$, we have
  $P_i(r_i(m)) = \go$.
  Suppose that $(r,m) \sim_i (r',m')$. Then  we must have
  $P_i(r'_i(m')) =  \go$.
  By validity, $\front_i(r',m') = 1$.
  By definition of $\sched$
  we have  $(\plan_i(r',m'),\plan_i(r',m'))) \in
    \sched'(r',m')$.
  Hence  $\I_{\gamma,P}, (r',m') \models \front_i \land (\lane_i,plan_i) \in \sched$.
  This shows that  $\I_{\gamma,P},(r,m)\models K_i(\front_i \land (\lane_i,\plan_i)\in \sigma ) $,
  as required.
\end{proof}

\begin{definition}[efficient intersection policies]
  An intersection policy $\sigma$ for a context $\gamma$ is \emph{efficient} if
  for all points $h \in \histories_\gamma$,
  we have that $\sigma(h)$ is a maximal conflict-free set of moves.
\end{definition}

\section{A Knowledge-Based Program with Lexicographically Optimal Implementations} \label{sec:optimal}

We would like to have a way to derive lexicographically optimal protocols under a range of failure assumptions.
Moreover, we want these protocols to be fair to all agents, even if there are agents present that are not.
To satisfy these goals, we start with an intersection policy $\sigma$
that can be run by all vehicles, including those
without V2V communications equipment.
One example of such $\sigma$ is the traffic light policy $\sigma^{TL}$.
In all cases, moves permitted by
this policy will have priority, but
we allow vehicles to violate the policy provided that they know that
they can do so safely. To avoid
clashes,
we
establish a priority order on the violations.
Let $\lead$ be a function from histories such that $\lead(h) \in \Lanesin$ for
each history $h$.
Intuitively, the agent at the
front of
the queue for lane $\lead(h)$ will get precedence in going through the
intersection
at the point $(r,m)$.
The context $\gamma$ is \emph{$\lead$-aware} if, for all protocols
$P$ for $\gamma$ and agents $i$ and $\ell \in \Lanesin$, we have
that $\I_{\gamma,P} \models \lead = \ell \rimp K_i(\lead = \ell)$.

\newcommand{\kbpnext}{\kbp}
Consider the following knowledge-based program $\kbpnext$,
where  $V_i$ is the proposition
\begin{quote}
  $(\lane_i,\plan_i) \not \in \sigma$ and the move $(\lane_i,\plan_i)$ is compatible with
  (a) all moves $(\lane_j,\plan_j)\in \sigma$ where $j$ is an agent
  who is about to enter the intersection (i.e., $\going _j$ holds)
  (b) all moves $(\lane_j,\plan_j)\not \in \sigma$ where $j \neq i$ is
  an agent for which $\going_j$ holds and $\lane_j \in [\lead,\lane_i)$.
  (Here
  $[\lead,\lane_i)$ is the set of lanes from $\lead(r,m)$ to
  $\lane_i \pmod{|\Lanesin|}$.)  \end{quote}
\begin{program}[H]
  \DontPrintSemicolon
  \lIf{$K_{i}(\front_i \land ((\lane_{i},\plan_{i}) \in \sigma \lor V_i))$}{$\go$}
  \lElse{$\noop$}
  \caption{$\kbpnext_i$}
\end{program}

Intuitively, this knowledge-based program allows all agents
permitted by $\sigma$ to go to do so,
as well as allowing agents not permitted by $\sigma$ to go, provided
they do so in a cyclic priority order, and each agent that goes knows
that its move is compatible with the moves of all agents of higher priority
(including agents permitted to go by $\sigma$).

\begin{propositionrep} \label{prop:optimal}
  Let $\sigma$ be a conflict-free intersection policy.
  If context $\gamma$ is synchronous, $\lead$-aware, and
  $\sigma$-aware,  then
  there exists a unique implementation
  $P$
  of $\kbpnext$
  that
  satisfies safety and validity,
  is lexicographically optimal
  with respect to $\gamma$,
  and lexicographically dominates the unique implementation of $\kbp^\sigma$.
  Moreover, if
  $\sigma$ is fair
  then $P$ satisfies liveness.
\end{propositionrep}
\begin{proof}
  We first show that there exists a protocol $P$ that implements $\kbpnext$ in $\gamma$.
  Since $\gamma$ is synchronous, we can define $P$ by induction on the
  time $\Time$ implicitly encoded
  in agents' local states. For all agents $i$ and all states $s\in L_i$
  with $\Time_i(s) = 0$, we define $P_i(s) = \noop$.
  Having defined $P_i$ on states up to time $m-1$, we have uniquely determined
  the set of all global states that occur at time $m$ on runs of $P$ in context $\gamma$.
  For agent $i$, we define $P_i$ on local states $s\in L_i$ with $\Time_i(s) = m$ as follows.

  For an agent $i$ such that $\front_i(s) = 0$, define $P_i(s) =
    \noop$.
  Note that for agents $i$ with
  $\front_i(s) = 1$,
  if $r_i(m) = s$, then
  $\I_{\gamma,P}, (r,m)) \models K_i(\front_i)$,
  since, by the minimal sensor model, $\front_i \dimp K_i \front_i$ is
  valid in $\I_{\gamma,P}$.
  For agents $i$ such that
  $\I_{\gamma,P}, r_i(m) \models K_i(\front_i \land (\lane_{i},\plan_{i}) \in \sigma)$,
  define  $P_i(s) = \go$.  Note that this is well defined, because  the
  formula $\front_i \land (\lane_{i},\plan_{i}) \in \sigma$ refers only
  to the past and present,
  and all runs of the protocol $P$ have been determined up to time $m$.
  Finally, we consider the remaining agents in the
  the cyclic order
  $\lead, \lead+1 \ldots , \lead -1 (\mod |\Lanesin|)$,
  and take $P_i(s) =
    \go$ if  $\I_{\gamma,P}, r_i(m) \models K_i(V_i)$ and $P_i(s) =
    \noop$ otherwise.
  This also is well defined, even
  though it refers to what the protocol $P$
  does at time $m$, because by synchrony and the fact that the behaviour
  of $P_i$ at all points of
  $\I_{\gamma,P}$ at time $m$ has been determined for all agents
  referenced by $V_i$.
  We use here the fact that $\gamma$ is $\sigma$-aware and
  $\lead$-aware, so each agent
  at the front of a queue
  can determine
  $\lead$ and $\sigma$ at each point.

  To show that protocol $P$ is an implementation of $\kbpnext$, we
  again proceed by induction on time.
  For $m=0$, we have $\front_i(r,0) = 0$ for all agents $i$ and runs $r$ of $P$, so
  $\I_{\gamma,P},(r,0) \models \neg K_{i}(\front_i \land ((\lane_{i},\plan_{i}) \in \sigma \lor V_i))$.
  Thus, $\kbpnext^{\I_{\gamma,P}}(r_i(0)) = \noop = P_i(r_i(0))$, as required.

  Suppose that for all runs $r$ of $P$ in context $\gamma$ and all $k< m$,
  we have  that
  $\kbpnext^{\I_{\gamma,P}}(r_i(k)) = P_i(r_i(k))$. Let $r$ be a run of $P$ in context $\gamma$.
  We show that  $\kbpnext^{\I_{\gamma,P}}(r_i(m)) = P_i(r_i(m))$.
  \begin{itemize}
    \item If $\front_i(r_i(m)) =0$,  then $P_i(r_i(m)) = \noop$ by
          definition.
          Necessarily, $\I_{\gamma,P}, r_i(m) \models \neg K_i(\front_i)$, so also
          $\I_{\gamma,P}, r_i(m) \models \neg K_i(\front_i \land ((\lane_{i},\plan_{i}) \in \sigma)\lor V_i))$,
          and we have that $\kbpnext^{\I_{\gamma,P}}(r_i(m)) = \noop = P_i(r_i(m)) $.
    \item If
          $\front_i(r_i(m)) =1$ and
          $\I_{\gamma,P}, r_i(m) \models K_i(\front_i \land
            (\lane_{i},\plan_{i}) \in \sigma)$, then
          $P_i(r_i(m)) = \go$. Clearly,
          $\I_{\gamma,P}, r_i(m) \models K_i(\front_i \land
            ((\lane_{i},\plan_{i}) \in \sigma)\lor V_i)$,
          so that  $\kbpnext^{\I_{\gamma,P}}(r_i(m)) = \go $.

    \item Finally, if $\front_i(r_i(m)) =1$ and
          $\I_{\gamma,P}, r_i(m) \models \neg K_i(\front_i \land
            (\lane_{i},\plan_{i}) \in \sigma)$,
          then since $\front_i$ is known by all agents and this context
          satisfies $\sigma$-awareness, we have that
          $\I_{\gamma,P}, r_i(m) \models K_i(\front_i \land (\lane_{i},\plan_{i}) \not \in \sigma)$.
          We show that $\kbpnext^{\I_{\gamma,P}}(r_i(m)) =
            P_i(r_i(m)) $
          by induction on the cyclic order.
          Assuming that we have this property for all such agents earlier in
          this order,
          then, by definition, $P_i(r_i(m)) = \go$ iff $\I_{\gamma,P}, r_i(m)
            \models K_i(V_i)$.
          If $\I_{\gamma,P}, r_i(m) \models K_i(V_i)$, then
          we also have that
          $\I_{\gamma,P}, r_i(m) \models K_i(\front_i \land ((\lane_{i},\plan_{i}) \in \sigma \lor V_i))$,
          hence $\kbpnext^{\I_{\gamma,P}}(r_i(m)) = \go = P_i(r_i(m))$, as required.
          On the other hand, if $\I_{\gamma,P}, r_i(m) \models \neg K_i(V_i)$, then
          $P_i(r_i(m)) = \noop$.  Moreover,
          there exists a run $r'$ of $P$ in context $\gamma$ such that $r'_i(m) =
            r_i(m)$ and
          $\I_{\gamma,P}, r'_i(m) \models \neg V_i$. Because
          $\I_{\gamma,P}, r_i(m) \models K_i(\front_i \land (\lane_{i},\plan_{i}) \not \in \sigma)$,
          it follows that
          $\I_{\gamma,P}, r'_i(m)
            \models (\lane_{i},\plan_{i}) \not \in \sigma) \land \neg V_i$,
          which implies that
          $\I_{\gamma,P}, r_i(m) \models \neg K_i(\front_i \land ( (\lane_{i},\plan_{i})\in \sigma \lor V_i)$.
          Hence $\kbpnext^{\I_{\gamma,P}}(r_i(m)) = \noop = P_i(r_i(m)) $.
  \end{itemize}
  This completes the proof that $P$ implements $\kbpnext$.

  To show that $P$ is the unique implementation of $\kbpnext$ in context
  $\gamma$, suppose that $P'$ is another implementation of $\kbpnext$.
  We show that, for all corresponding runs $r$ of $P$
  and $r'$ of $P'$ in $\gamma$ and agents $i$, we have that
  $P_i(r_i(m)) = P'_i(r'_i(m))$
  and that
  $\I_{\gamma,P}, r_i(m) \models K_i(\front_i \land ( (\lane_{i},\plan_{i})\in \sigma \lor V_i)$ iff
  $\I_{\gamma,P'}, r_i(m) \models K_i(\front_i \land ( (\lane_{i},\plan_{i})\in \sigma \lor V_i)$
  by induction on $m$.
  For $m =0$, this follows from the fact that there can be no agents at the front of any queue,
  so $P_i(r_i(0)) = \noop = \kbpnext^{\I_{\gamma,P'}}(r_i(0)) = P'_i(r'_i(m))$.
  For the induction step, suppose that  $P_i(r_i(k)) = P'_i(r'_i(k))$ for all
  $k<m$ and all corresponding runs $r$ of $P$
  and $r'$ of $P'$ in $\gamma$. This means that
  $r[0\ldots m] = r'[0\ldots m]$
  for all corresponding runs $r$ and $r'$, and, in particular that $r_i(m) =
    r'_i(m)$.
  The fact that $P_i(r_i(m)) = P'_i(r'_i(m))$ now follows from the fact that
  $P_i(r_i(m))$ and $P'_i(r'_i(m))$ are identically determined from the satisfaction of
  the formula $K_i(\front_i \land ( (\lane_{i},\plan_{i})\in \sigma \lor V_i)$, by an
  induction over, first, the agents that are not at the front of a queue,
  then, the lanes $\ell$ such that there exists an agent $i$ with
  $\front_i(r,m) = \front_i(r',m) = 1$  and $(\lane_i, \plan_i) \in \sigma(H(r,m)) = \sigma(H(r',m))$,
  and next, the lanes
  $\lead, \ldots , \lead-1 (\mod |\Lanesin|) $ in
  the cyclic order that have an
  agent at the front at
  $(r,m)$ (equivalently, at $(r',m)$).

  Next, we show that $P$ is an intersection protocol.
  \begin{itemize}

    \item For validity, suppose that $\I_{\gamma,P},(r,m) \models \going_i$.
          Then $P_i(r_i(m)) = \go$ and
          $\I_{\gamma,P}, (r,m) \models K_i(\front_i \land ( (\lane_{i},\plan_{i})\in \sigma \lor V_i)$,
          from which it follows that $\I_{\gamma,P}, (r,m) \models \front_i$, as required.

    \item  For safety, suppose that
          $\I_{\gamma,P}, (r,m) \models \going_i \land \going_j$. (By symmetry of $\noconflict$, it suffices to consider
          the case $i\neq j$.)
          This means  $\I_{\gamma,P}, (r,m) \models K_i(\front_i \land ( (\lane_{i},\plan_{i})\in \sigma \lor V_i)$
          and $\I_{\gamma,P}, (r,m) \models K_j(\front_j \land ( (\lane_{j},\plan_{j})\in \sigma \lor V_j)$.
          We consider the following cases.
          \begin{itemize}
            \item Both $(\lane_i(r,m),\plan_i(r,m))$ and $(\lane_j(r,m),\plan_j(r,m))$
                  are in $\sigma(r,m)$.
                  In this case, these moves are compatible, since $\sigma$ is
                  conflict-free.
            \item One of  $(\lane_i(r,m),\plan_i(r,m))$ and $(\lane_j(r,m),\plan_j(r,m))$
                  is in $\sigma(r,m)$ and the other is not.  Without loss of
                  generality, suppose that $(\lane_i(r,m),\plan_i(r,m))\in \sigma(r,m)$,
                  Then we have  that $\I_{\gamma,P}, (r,m) \models  V_j$,
                  which implies that $(\lane_j(r,m),\allowbreak\plan_j(r,m))$ is compatible with
                  $(\lane_i(r,m),\plan_i(r,m))$ by case (a) of the definition of $V_j$.
            \item Neither $(\lane_i(r,m),\plan_i(r,m))$ nor $(\lane_j(r,m),\plan_j(r,m))$
                  is in $\sigma(r,m)$.
                  Without loss of generality, suppose that $\lane_i(r,m)$ precedes
                  $\lane_j(r,m)$ in the
                  cyclic sequence.
                  By $\sigma$-awareness and $\lead$-awareness, we have
                  that
                  $\I_{\gamma,P}, (r,m) \models K_j(\neg  (\lane_{i}, \allowbreak \plan_{i})\in \sigma)$,
                  which implies that
                  $\I_{\gamma,P}, (r,m) \models K_j(V_j)$, from which it follows that
                  $\I_{\gamma,P}, (r,m) \models V_j$. By part (b) of the definition of $V_j$,
                  it follows that $(\lane_j(r,m),\allowbreak\plan_j(r,m))$ is compatible with $(\lane_i(r,m),\plan_i(r,m))$.
          \end{itemize}
          Thus $(\lane_j(r,m),\plan_j(r,m))$ is compatible with $(\lane_i(r,m),\plan_i(r,m))$ in each case,
          showing that $P$ satisfies safety.

    \item For liveness, we prove that $
            \I_{\gamma,P},(r,m) \models  \front_i \rimp
            \Diamond \going_i$.
          Suppose, by way of contradiction, that $\I_{\gamma,P},(r,m) \models
            \front_i \land \Box \neg \going_i$.
          Since $\sigma$ is fair, there exists $k\geq m$ such that
          $(\lane_i(r,k), \plan_i(r,k)) \in \sigma(r,k)$.
          Because agent $i$ does not go after time $m$, we have that $\front_i(r,m)
            = 1$ and that $(\lane_i(r,k), \plan_i(r,k)) = (\lane_i(r,m),\plan_i(r,m))$.
          By  $\sigma$-awareness of $\gamma$, and the fact that the agent knows the values of $\front_i$, $\lane_i$ and $\plan_i$, since these are encoded in the local state,
          it follows
          that  $\I_{\gamma,P},(r,k) \models K_i(\front_i \land (\lane_i,\plan_i)\in \sigma)$.
          We obtain from this that $P_i(r_i(k)) = \go$, contradicting  $\I_{\gamma,P},(r,k) \models \neg \going_i$.
  \end{itemize}

  To show that $P$ is lexicographically optimal, suppose by way of
  contradiction that there
  exists an intersection protocol $P'$
  that strictly lexicographically dominates $P$.
  Let $m$ be the earliest time for which there exists a run $r$ of $P$
  such that for the corresponding run $r'$ of $P'$, we have $\GO(r',m) \supsetneq \GO(r,m)$, and
  let $i$ be an agent in $\GO(r',m) \setminus \GO(r,m)$.
  Since $\GO(r',k) = GO(r,k)$ for $k<m$, we have that $r[0\ldots m] =
    r'[0 \ldots m]$.
  (In fact, this holds for all pairs of corresponding runs of $P$ and
  $P'$ and time $m$.)
  Because $i \in \GO(r',m)$, we have that $1 = \front_i(r',m) =
    \front_i(r,m)$.
  Since $i \not \in \GO(r,m)$, we have
  that
  $\I_{\gamma,P}, (r,m) \models K_i(\front_i \land ( (\lane_{i},\plan_{i})\in \sigma \lor V_i)$.
  It follows that there exists a run $r^1$ of $P$ with $r_i(m) = r^1_i(m)$ such that
  $\I_{\gamma,P}, (r^1,m) \models \neg( \front_i \land ( (\lane_{i},\plan_{i})\in \sigma \lor V_i))$.
  By the fact that $\front_i$ is encoded into the agent $i$'s local state,
  $\front_i(r^1,m) = \front_i(r,m)$, so
  $\I_{\gamma,P}, (r^1,m) \models ( (\lane_{i},\plan_{i})\not \in \sigma \land \neg V_i))$.
  This implies that there exists an agent $j$ that precedes $i$ in the
  cyclic sequence such that
  $P_j(r_j(m)) = \go$ and that $(\lane_j(r^1,m),\plan_j(r^1,m))$ is
  incompatible with
  $(\lane_i(r^1,m),\plan_i(r^1,m))$. Because both $\lane_i$ and $\plan_i$ are components of the sensor reading
  in agent $i$'s local state, we have that $(\lane_i(r^1,m),\plan_i(r^1,m))
    \allowbreak =  (\lane_i(r,m),\plan_i(r,m))$.
  Thus, $(\lane_j(r^1,m),\allowbreak\plan_j(r^1,m))$ is incompatible with
  $(\lane_i(r,m),\allowbreak\plan_i(r,m))$. We also have $j \in \GO(r^1,m)$.

  Let $r^{1'}$ be the run of $P'$  corresponding to $r^1$. By the
  observation above, we have  that
  $r^1[0\ldots m] = r^{1'}[0\ldots m]$ and that $\GO(r^1,m) \subseteq
    \GO(r^{1'},m)$.
  In particular, we obtain that $j \in \GO(r^{1'},m)$.
  Since $r^1[0\ldots m] = r^{1'}[0\ldots m]$, it follows that
  $(\lane_j(r^1,m), \allowbreak \plan_j(r^1,m)) \allowbreak = \allowbreak (\lane_j(r^{1'},m),\plan_j(r^{1'},\allowbreak m))$.
  Moreover, since $r[0\ldots m] \allowbreak = r'[0 \ldots m]$,
  we have that $(\lane_i(r,m),\allowbreak \plan_i(r,m)) \allowbreak = (\lane_i(r',m), \allowbreak \plan_i(r',m))$.
  Also, $r^{1'}_i(m) =  r^{1}_i(m) =  r_i(m) =  r'_i(m)$.
  Because $\lane_i$ and $\plan_i$ are components of agent $i$'s local state,
  these variables have the same values at these four points.
  Since $i \in \GO(r',m)$, we have that $i \in \GO(r^{1'}, m)$.
  But now, we have at point $(r^{1'}, m)$, that agent $i$ makes move
  $(\lane_i(r,m),\plan_i(r,m))$ which we know to be incompatible with the move
  \[(\lane_j(r^1,m),\plan_j(r^1,m)) = (\lane_j(r^{1'},m),\plan_j(r^{1'},m))\]
  made by agent $j$ at this point.  This contradicts the Safety of protocol $P'$,
  so $P'$ is not an intersection protocol as assumed.
  This completes the proof that $P$ is lexicographically optimal.

  Finally, to show that $P$ dominates the implementation of $\kbp^\sigma$ in $\gamma$,
  suppose that $P'$ is such an implementation. Consider an adversary with respect to which the
  behaviour of $P$ and $P'$ differ, let $r$ and $r'$ respectively be the runs of $P$ and $P'$
  with this adversary, and let $m$ be the minimal time for which $P$ and $P'$ differ in round $m+1$.
  We have $r[0\ldots m] = r'[0\ldots m]$,  so also $H(r,m) = H(r',m)$.

  For every agent $j\in \GO(r',m)$,  since $P'$ implements
  $\kbp^\sigma$, we have that $\I_{\gamma,P'},(r',m) \models K_j(\front_j \land
    (\lane_j,\plan_j) \in \sigma)$.
  Thus, we have that $\front_j(r',m) = 1$ and $(\lane_j(r',m),\allowbreak\plan_j(r'm)) \in
    \sigma(r',m) = \sigma(r',m)$. Because $\front_j$, $\lane_j$ and $\plan_j$ are
  encoded in agent $j$'s local state and $r(m) = r'(m)$, it follows that
  $\front_j(r,m) = 1$ and $(\lane_j(r',m),\allowbreak\plan_j(r'm)) \in \sigma(r',m) =
    \sigma(r',m)$. Using $\sigma$-awareness of $\gamma$, we obtain that
  $\I_{\gamma,P},(r,m) \allowbreak \models \allowbreak
    K_j(\front_j \land \lane_j,\plan_j) \in \sigma)$,
  But this implies that $\I_{\gamma,P},(r,m) \models K_j(\front_j \allowbreak
    \land \allowbreak
    ((\lane_j,\plan_j) \allowbreak \in \sigma) \lor V_j)$, so also $j \in \GO(r,m)$. It
  follows that  $\GO(r',m) \subseteq \GO(r,m)$.
  Since these sets are different, by assumption, we must have that $\GO(r',m)
    \subsetneq \GO(r,m)$. This shows that $P$ lexicographically dominates $P'$.
\end{proof}

We can also obtain liveness of the  implementations of
$\kbp$
under some other conditions.
Define the function $\lead$ to be \emph{fair} if, for all feasible
sequences of histories
$h_0, h_1, \ldots$, all $m\geq 0$ and all lanes $\ell \in \Lanesin$,
there exists $m'\geq m$ such that $\lead(h_{m'}) = \ell$.  Intuitively, fairness of $\lead$ will
ensure that $\lead$ fairly selects the first agent that
can violate the intersection  policy according to
$\kbp$
when this can
be done safely.

We also need to ensure that it is not the case that $\sigma$
always gives priority to
other lanes whenever the lane $\ell$ is selected by $\lead$. For this,
define a pair $(\sigma,\lead)$, consisting of an intersection policy
$\sigma$ and a function $\lead$,
to be \emph{fair} if for all feasible sequences of histories
$h_0, h_1, \ldots$, all $m\geq 0$ and all moves $(\ell,\ell') \in \Lanesin$,
there exists $m'\geq m$ such that either $(\ell,\ell')\in \sigma(h_{m'})$,
or $\lead(h_{m'}) = \ell$ and
$(\ell,\ell')$ is compatible with all the moves in $\sigma(h_{m'})$

\begin{propositionrep} \label{propn:sigmaleadlive}
  Let $P$ be an implementation of
  $\kbp$
  with respect to a synchronous,
  $\lead$-aware and $\sigma$-aware context. If the pair $(\sigma,\lead)$ is fair,
  then $P$ satisfies liveness.
\end{propositionrep}
\begin{proof}
  The argument is similar to that of the proof of liveness based on
  fairness of $\sigma$ in
  Proposition~\ref{prop:optimal}.
  We prove that $
    \I_{\gamma,P},(r,m) \models  \front_i \rimp
    \Diamond \going_i$.
  Suppose, by way of contradiction, that $\I_{\gamma,P},(r,m) \models
    \front_i \land \Box \neg \going_i$.
  Since $(\sigma,\lead)$ is fair, there exists $k\geq m$ such that either
  $(\lane_i(r,k), \plan_i(r,k)) \in \sigma(H(r,k))$, or
  $\lead(H(r,k)) = \lane_i(r,k)$ and
  $(\lane_i(r,k), \plan_i(r,k))$ is compatible with all the moves in $
    \sigma(H(r,k))$.

  Because agent $i$ does not go after time $m$, we have that $\front_i(r,m)
    = 1$ and that $(\lane_i(r,k), \allowbreak \plan_i(r,k)) = (\lane_i(r,m), \plan_i(r,m))$.
  Since $\gamma$ is  $\sigma$-aware and the agent
  knows the values of $\front_i$, $\lane_i$ and $\plan_i$ (which
  are encoded in its local state),
  it follows
  that  either $\I_{\gamma,P},(r,k) \models K_i(\front_i \land
    (\lane_i,\plan_i)\in \sigma)$
  or $\I_{\gamma,P},(r,k) \models K_i(\front_i \land V_i)$.
  We obtain from this that $P_i(r_i(k)) = \go$, contradicting  $\I_{\gamma,P},(r,k) \models \neg \going_i$.
\end{proof}

Note that if $\lead$ is fair, and the $\sigma_\emptyset$ is the (unfair) intersection policy defined by
$\sigma_\emptyset(h) = \emptyset$ for all histories $h$, then the pair $(\sigma_\emptyset,\lead)$ is fair.
For examples in which $\sigma$ is not trivial, consider the following properties of $\sigma$.
Say that  $\sigma$ is \emph{cyclic} (with cycle length $k$)  if for all histories
$h$ and $h'$  with $|h| \equiv |h'| \mod k$, we have $\sigma(h) = \sigma(h')$.
Say that $\sigma$ is \emph{non-excluding} if for all moves $(\ell, \ell')$,
there exists a history $h$ such that
$(\ell, \ell')$ is compatible with all moves in $\sigma(h)$.
Given a non-excluding $\sigma$ with cycle length $k$,
let $\lead$ be defined by $\lead(h) = \lfloor h/k\rfloor \mod k$. Then $(\sigma,\lead)$ is fair.
This is because the value of $\lead$ cycles through all values in $\Lanesin$,
but is held constant through each cycle of $\sigma$. Thus, for each move $(\ell, \ell')$,
eventually a point in these combined cycles will be reached for which the value of $\lead$ is $\ell$ and
$(\ell, \ell')$ is compatible with all moves permitted by $\sigma$.

\subsection{Implementing $\kbp$ when there is no communication}
\label{sec:impl-no-comm}
We now consider standard implementations of $\kbp$ in two particular
contexts of interest. Since we would like the implementations to be correct and
lexicographically optimal, we use $\lead$ and $\sigma$
defined as $\lead(h) = m \mod |\Lanesin|$ and
$\sigma(h) = \emptyset$ for all histories $h$ of length $m$.
Using this choice of $\lead$ and $\sigma$ in the construction of $\kbp$
ensures that in any synchronous intersection context, both
$\lead$\emph{-awareness} and $\sigma$\emph{-awareness} hold; moreover,
the pair $(\lead,\sigma)$ is fair.
Therefore, implementations $P$ of $\kbp$ in such contexts
are correct and lexicographically
optimal, by  Propositions~\ref{prop:optimal} and \ref{propn:sigmaleadlive}.

We have taken $\sigma$ to be empty for ease of exposition.
For practical implementations, the construction given
by the proof of Proposition~\ref{prop:optimal} can be used to get other
implementations that prioritize moves permitted by $\sigma$.
(For example, in an
intersection where certain lanes are often busier, moves originating from those
lanes can be prioritized.)
Note that
for empty $\sigma$,
the condition $K_i(\front_i \land ((\lane_i,\plan_i)
  \in \sigma \lor V_i))$ reduces to
$K_i(\front_i \land V_i')$, where $V_i'$ is the proposition
\begin{quote}
  ``the move $(\lane_i,\plan_i)$ is compatible with all moves
  $(\lane_j,\plan_j)$
  of agents $j \neq i$ with $\lane_j \in [\lead, \lane_i)$
  such that $\going_j$'',
\end{quote}
since $\sigma$ is empty.
Consider the following context with no communication.
Let $\gamma_{\emptyset}$ be a synchronous intersection context where agents do
not broadcast messages. Formally,
for a failure model $\F$, we define
$\gamma_\emptyset(\F) =  (\exchange_\emptyset, \F, \pi_\emptyset)$, where
\begin{itemize}
  \item $(\exchange_\emptyset)_i$ is an
        information-exchange protocol
        where the following hold:
        \begin{itemize}
          \item The
                set of memory states is a singleton so, effectively,
                local states consist only of the sensor reading
                $L_i = \Sigma_i$.
          \item No messages are sent, so $M_i = \emptyset$,
                $\mu_i$ is the
                constant function with value $\bot$, and $\delta_i$ is omitted.
          \item The sensor model is defined as in the
                definition of intersection contexts. The only modification is
                that the sensor model now maps environment states to tuples of
                the form $\langle \front_i, \lane_i, \plan_i,
                  \mathit{time}_i\rangle$, where $\mathit{time}_i$ is determined by
                the time
                encoded in the environment state.
        \end{itemize}
  \item $\pi_\emptyset$ interprets the propositions defined for
        intersection
        contexts in the obvious way.
\end{itemize}

We now define a procedure to compute a set $\poss_i$  of moves that
agent $i$ believes may be performed as a function of
$\lead$ and the structure of the intersection represented by
$\noconflict$.
We capture stages of the
construction of this
set as sets of moves $\poss_I^l$ for $l \in  [\lead-1, \lane_i)$.
(By $\lead$-awareness, $\lead$ is computable from the agent's local state. For
brevity, we interpret $\lead-1$ as $\lead -1 (\mod |\Lanesin|)$.)
\begin{enumerate}
  \item Start with $\poss_i = \poss_i^{\lead-1} = \emptyset$
  \item For $l \in [\lead, \lane_i)$ do
        \begin{enumerate}
          \item Let $L$ be the set of moves $(l,l')$ where $l'\in \Lanesout$ such that
                $(l,l')$ is compatible with $\poss_i$, and let
                $\poss_i^l := \poss_i \cup L$ and $\poss_i := \poss_i^l$.
        \end{enumerate}
  \item Output $\poss_i$.
\end{enumerate}

Let $P^\emptyset$ be the standard protocol given by the following program,
where
move $(l,l')$ is compatible with a set of moves $S$ if it is compatible
with all moves in $S$ according to $\noconflict$.

\begin{program}[H]
  \DontPrintSemicolon
  \lIf{$\front_i \land
      (\lane_i,\plan_i) \text{ is compatible with } \poss_i$}{$\go$}
  \lElse{$\noop$}
  \caption{$P^\emptyset_i$}
\end{program}
\begin{propositionrep}\label{pro:impelement1}
  $P^\emptyset$ implements $\kbp$ with respect to $\gamma_\emptyset(\F)$
  for $\F \in \{\nofail, \CR, \SO\}$.
\end{propositionrep}
\begin{proof}
  Let $\I = \I_{\gamma_\emptyset(\F),\kbp}$.
  We show that
  $P_i^\emptyset(r_i(m)) = \kbp_i^\I(r_i(m))$.
  We consider three cases.
  \begin{itemize}
    \item If $\I,(r,m) \models \neg\front_i$, then we have
          $P_i^\emptyset(r_i(m)) = \kbp_i^\I(r_i(m)) = \noop$.
    \item If $\I,(r,m) \models \front_i \land (\lane_i,\plan_i) \text{ is
              incompatible with } \poss_i$,
          then $P_i^\emptyset(r_i(m)) = \noop$.
          Now consider
          $(l,l') \in \poss_i(r,m)$
          such that the move $(l,l')$ is incompatible
          with $(\lane_i(r,m),\allowbreak\plan_i(r,m))$.
          Let $j \neq i$ be any agent.
          Observe that by the definition of $\poss_i(r,m)$,
          $l \in [\lead, \lane_i)$.
          Since there is no message exchange,
          there exists a point $(r',m)$ where this agent $j$ is waiting in
          front of lane $l$, and has $\plan_j(r',m) = l'$, such that
          $(r,m) \sim_i (r',m)$. Hence,
          $\I,(r,m) \models \neg K_i(V_i')$, so
          $\kbp_i^\I(r_i(m)) = \noop$.
    \item If $\I,(r,m) \models \front_i \land (\lane_i,\plan_i) \text{ is
              compatible with } \poss_i$,
          then $P_i^\emptyset(r_i(m)) = \go$.
          Let $(r,m) \sim_i (r',m')$. Then, by the definition of the sensor readings, we have
          $m = \Time_i(r',m) =\Time_i(r,m') = m'$,
          $\front_i(r',m) =\front_i(r,m)$,
          $\lane_i(r',m) =\front_i(r,m)$, and   $\plan_i(r',m) =\plan_i(r,m)$.
          Since $\first_i$ depends only on  these variables, and  $\I,(r,m) \models \first_i$,
          we have $\I,(r',m) \models \first_i$.
          The proposition $V_i'$ must hold at $(r',m)$, since
          at $(r',m)$ there can be no agent $j$ with $\lane_j$ in $[\lead, \plan_i)$
          whose move can conflict with $i$'s move.
          Hence, for all points $(r',m) \sim_i (r,m)$, $\I,(r',m) \models \front_i \land V_i'$,
          so $\I(r,m) \models  K_i(\front_i \land V_i')$ and  $\kbp_i^\I(r_i(m)) = \go$.
  \end{itemize}
  Therefore, $P^\emptyset$ implements $\kbp$ in $\gamma_\emptyset$.
\end{proof}

Proposition~\ref{pro:impelement1} shows that, without communication, a
protocol that essentially
implements traffic lights is lexicographically optimal.

\subsection{Implementing $\kbp$ in a context with limited communication}
If we allow messages regarding the current lane and agents' intentions
by agents that reach the front,
this
changes how implementations of $\kbp$ behave.
Roughly speaking, in runs where the intersection gets crowded, a much
larger set of agents can proceed through the intersection.
Let $\gamma_{\plan}$ be a synchronous context with communication
failures such that if an agent is in the front of some lane, it broadcasts
$(\lane,\plan)$.
(This information exchange broadcasts a lot less information than a full-information exchange.)
More formally,\footnote{This context satisfies the
  \emph{sufficiently rich} condition of Section~\ref{sec:unnecessary-waiting}.}
for a failure model $\F$, we define $\gamma_\plan(\F) =  (\exchange_\plan, \F, \pi_\plan)$, where
\begin{itemize}
  \item $(\exchange_\plan)_i$ is defined as an information-exchange protocol
        where the following hold:
        \begin{itemize}
          \item The local states maintain a set of moves $M$
                in the memory component  in addition to the sensor readings. Intuitively, this set
                represents the set of moves from broadcasts that were
                received by $i$ in the current round.
                Note that $M_i$ may not contain $i$'s move
                since $i$'s broadcast may fail.
          \item
                The set of messages is $M_i =  \Lanesin\times \Lanesout$, and
                $\mu_i$  broadcasts
                the message $(\lane_i,\plan_i)$ by
                reading $\lane_i$ and $\plan_i$ from the
                sensor reading,
                if $\front_i$, and broadcasts no message otherwise. Note that these variable references are from
                $\sense(s_e')$ where $s_e'$ is the new environment state that the system moves to
                in the course of the round.
          \item The sensor model is defined as in the
                definition of intersection contexts
                (while including $\mathit{time}$ as a sensor reading as
                in $\exchange_\emptyset$).
          \item $\delta_i$ maps the set of received messages directly into
                the memory with $i$'s own move; that is, $\delta_i(s_i, a,
                  \mathit{Mes}) = \mathit{Mes}$.
                Note that an agent can determine from this set whether its own broadcast was successful.
        \end{itemize}
  \item $\pi_\plan$ interprets the propositions defined for interpretation
        contexts in the obvious way.
\end{itemize}

We now
proceed as in Subsection~\ref{sec:impl-no-comm} and define a procedure to
compute from an agent $i$'s local state $s_i=(M_i,(\lane_i,\plan_i,\Time_i))$ a
set $\poss_i$  of moves that agent $i$ believes may be performed by
higher-priority agents in the next round.
We again
capture stages of the construction of this set as sets of moves $\poss_I^l$ for
$l \in  [\lead-1, \lane_i)$.
\begin{enumerate}
  \item Start with $\poss_i = \poss_i^{\lead-1} = \emptyset$
  \item For $l \in  [\lead, \lane_i)$ do
        \begin{enumerate}
          \item If for some $l' \in \Lanesout$, the move $(l,l')$ is in $M_i$
                then \\
                \hspace*{2mm} if $(l,l')$ is compatible with $\poss_i$\\
                \hspace*{2mm} then $\poss_i^l := \poss_i \cup \{(l,l')\}$ and $\poss_i
                  := \poss_i^l$\\
                \hspace*{2mm} else $\poss_i^l := \poss_i$.
          \item Otherwise, let $L$ be the set of moves $(l,l')$ where $l'\in \Lanesout$
                such that $(l,l')$ is compatible with
                $\poss_i$, and let
                $\poss_i^l := \poss_i \cup L$ and $\poss_i := \poss_i^l$.\footnote{Intuitively, since $M$ is the set of
                  moves that $i$ hears about from agents in the front of some lane,
                  in this case $i$ did not hear from anyone in lane $l$. However,
                  in settings with sending omissions,
                  there may nevertheless be an agent at the front of lane $l''$.
                  Such an agent will move only if it can do so safely.}
        \end{enumerate}
  \item Output $\poss_i$.
\end{enumerate}
Let the output of running this procedure on a local state with memory
state $M_i$ be denoted by $\poss_i$, and
let $P^\plan$ be the standard protocol defined using the following program: \\
\begin{program}[H]
  \DontPrintSemicolon
  \lIf{$\front_i \land (\lane_i,\plan_i) \text{ is compatible with } \poss_i$}{$\go$}
  \lElse{$\noop$}
  \caption{$P^\plan_i$}
\end{program}

\begin{toappendix}
  \begin{lemma} \label{lem:sameposs}
    For $\F \in \{CR, SO\}$, suppose that $(r,m)$ is a point of $P^\plan$ in context $\gamma_\plan(\F)$,
    and let $i$ and $j$ be agents with $ \front_i(r,m) = \front_j(r,m) = 1 $ and
    $\lane_j(r,m) \in [\lead(r,m),\lane_i(r,m))$. Then $M_i(r,m) = M_j(r,m)$,
    and
    if $l \in [\lead(r,m),\allowbreak\lane_j(r,m))$ then $\poss_j^l(r,m) = \poss_i^l(r,m)$.
  \end{lemma}
  \begin{proof}
    By a straightforward induction on the construction of the sets $\poss_j^l(r,m)$ and $\poss_i^l(r,m)$.
    Fix agents $i$ and $j$.
    Note that because
    \begin{itemize}
      \item in $\gamma_\plan(\F)$, only agents at the front of the queue for a lane after a round
            broadcast during that round, and
      \item
            we assume that the transmission environment $T$ has the property that messages sent from
            the front of the queue for a lane are (barring failures)  received by all agents also  at the
            front of the queue for some lane, and
      \item in the failure models $CR$ and $SO$,
            there are no reception failures, so
            all broadcasts are atomic: either received by all agents at the front or received by none,
    \end{itemize}
    we have that all agents at the front receive the same set of messages, which
    become their memory state after the round, so that $M_i(r,m) = M_j(r,m)$.
    Agents $i$ and $j$ also compute the same value for $\lead$. This implies that
    if $l \in
      [\lead(r,m),\lane_j(r,m))$ then $\poss_j^l(r,m) = \poss_i^l(r,m)$.
  \end{proof}

  \begin{lemma}\label{lem:Pplan-poss}
    For $\F \in \{CR, SO\}$, suppose that $(r,m)$ is a point of $P^\plan$ in context
    $\gamma_\plan(\F)$, and let $i$ be an agent with $\front_i(r,m) = 1$
    and $(l,l')$ a move with $l \in [\lead(r,m),\allowbreak\lane_i(r,m))$.
    Then $(l,l') \in \poss_i(r,m)$ iff $\I_{P^\plan,\gamma_\plan(\F)}, \allowbreak
      (r,m) \models \neg K_i\neg ( \exists j\in \Agents(
        \front_j \land \lane_j = l \land \plan_j = l' \land \going_j))$.
  \end{lemma}
  \begin{proof}
    Let $\I = \I_{P^\plan,\gamma_\plan(\F)}$. We first suppose that $(l,l') \in
      \poss_i(r,m)$ and show that $\I,(r,m) \models \neg K_i\neg ( \exists j\in
      \Agents(\front_j \land \lane_j = l \land \plan_j = l' \land \going_j))$.
    Let $r'$ be the run constructed as follows. If the adversary in run $r$ is
    $\adversary = (\arrival, \canhear, F_t, F_r)$, the adversary $\adversary' =
      (\arrival', \canhear, F'_t, F_r)$ in run $r'$ satisfies the following. The
    transmission environment $\canhear$ and the reception failures $F_r$ are
    identical to those in $\adversary$. With respect to arrivals
    $\arrival'$, no agents arrive at the intersection until round $m$. In round $m$,
    we have the following arrivals:
    \begin{itemize}
      \item Agent $i$ arrives in lane $\lane_i(r,m)$ with intent
            $\plan_i(r,m)$.
      \item If $j\neq i$ is an agent with $(\lane_j(r,m),\plan_j(r,m)) \in
              M_i(r,m)$, then agent $j$ arrives in lane $\lane_j(r,m)$ and intent
            $\plan_j(r,m)$.
      \item If $(l,l')\not \in M_i(r,m)$, then some other agent $i^*$ arrives in
            lane $l$ with intent $l'$.
    \end{itemize}
    There are no other arrivals in round $m$. After round $m$, arrivals can be
    scheduled arbitrarily. With respect to transmission failures, for agents $j$
    such that  $(\lane_j(r,m),plan_j(r,m)) \in M_i(r,m)$, we have $F'_t(m-1,j) = 1$,
    and if the special agent $i^*$ arrives in round $m$ of $r'$, then
    $F'_t(m-1,i^*)
      = 0$. All other aspects of $\adversary'$ can be defined arbitrarily. It is
    straightforward to show that $(r,m) \sim_i(r',m)$. Moreover, by construction, we
    have $M_j(r,m) = M_j(r',m)$ for all agents $j$ with $\front_j(r,m) = 1$.
    Indeed, $r_j(m) = r'_j(m)$.
    We show that $\I,(r',m) \models   \exists j\in \Agents(\front_j \land \lane_j =
      l \land \plan_j = l' \land \going_j)$.

    There are two cases. First, consider the case where $(l,l') \in M_i(r,m)$. In
    this case, there exists an agent $j \neq i$ such that $\I,(r,m) \models \front_j
      \land \lane_j = l \land \plan_j = l'$, hence also $\I,(r',m) \models \front_j
      \land \lane_j = l \land \plan_j = l'$. We also have
    that
    either $l = \lead(r,m)$ or
    $(l,l')$ is compatible with $\poss_i^{l-1}(r,m)$. If $l = \lead(r,m)$ then
    $\poss_j(r,m) = \emptyset$, and it is immediate that $(l,l')$ is compatible with
    $\poss_j(r,m)$, so $P^\plan_j(r_j(m)) = \go$ and $\I(r',m) \models \going_j$.
    On the other hand, if $(l,l')$ is compatible with $\poss_i^{l-1}(r,m)$, then by
    Lemma~\ref{lem:sameposs}, we have that $\poss_i^{l-1}(r,m)= \poss_j^{l-1}(r,m)$.
    Since $r_j(m) = r'_j(m)$, we also have $\poss_j^{l-1}(r,m)= \poss_j^{l-1}(r',m)=
      \poss_j(r',m)$. Thus, $(l,l')$ is compatible with $\poss_j(r,m)$, and it follows
    again that $P^\plan_j(r_j(m)) = \go$  and $\I,(r',m) \models \going_j$. Hence
    $j$ provides the required witness.

    Second, consider the case where $(l,l') \not \in M_i(r,m)$. In this
    case, we have
    $\I,(r',m) \models \front_{i^*} \land \lane_{i^*} = l \land \plan_{i^*} = l'$.
    Also $F'_t(m-1,i^*) =0$, and no other agent can be at the front of
    lane $l$. By definition of $\gamma_\plan(\F)$, no other agent can
    broadcast a move from lane $l$, so we do not have $(l,l'') \in M_i(r,m)$ for any
    $l''$. The reason that $(l,l') \in \poss_i(r,m)$ must therefore be case 2(a) of
    the construction. If $l = \lead(r'm)$ then $\poss_{i^*}(r',m) = \emptyset$, and
    it is immediate that $(l,l')$ is compatible with $\poss_{i^*}(r',m)$. It follows
    that $\I(r',m) \models \going_{i^*}$. If $l \neq \lead(r',m)= \lead(r,m)$ then
    $(l,l') \in \poss_i(r,m)$ because $(l,l')$ is compatible with
    $\poss_{i}^{l-1}(r,m)= \poss_{i}^{l-1}(r',m)$.
    By Lemma~\ref{lem:sameposs}, we have that $\poss_i^{l-1}(r',m)=
      \poss_{i^*}^{l-1}(r',m)$. Because $\lane_{i^*}(r',m) = l$, we furthermore have
    $\poss_{i^*}^{l-1}(r',m) = \poss_{i^*}(r',m)$. Thus, $(l,l')$ is compatible
    with $\poss_{i^*}(r',m)$, and it follows that $P^\plan_{i^*}(r_{i^*}(m)) = \go$
    and $\I(r',m) \models \going_{i^*}$. Hence $j=i^*$ provides the required witness.
    This completes the proof of the implication from left to right.

    For the converse, assume that  $\I,(r,m) \models \neg K_i\neg ( \exists j\in
      \Agents(\front_j \land \lane_j = l \land \plan_j = l' \land
        \going_j))$.  We show that
    $(l,l') \in \poss_i(r,m)$. There exists a run $r'$ of $P^\plan$
    such that $(r,m) \sim_i (r',m)$ and $\I,(r',m) \models \front_j \land \lane_j =
      l \land \plan_j = l' \land \going_j$ for some agent $j$. By definition of local
    states, we have $\lane_i(r,m) =\lane_i(r',m)$, $\plan_i(r,m) = \plan_i(r',m)$,
    and $M_i(r,m) = M_i(r',m)$. Suppose that $(l,l') \not \in \poss_i(r,m) =
      \poss_i(r',m)$. We cannot have $l = \lead(r',m)$, otherwise $(l, l')$
    would be compatible with $\poss_i(r',m) = \emptyset$, and we would
    have $(l,l') \in
      \poss_i(r,m) =
      \poss_i(r',m)$, contrary to assumption. It follows (both in the case that
    $(l,l') \in M_i(r',m)$ and the case that $(l,l') \not \in M_i(r',m)$, for which
    there can be no $l''$ with $(l,l'') \in M_i(r',m)$, by the assumptions on
    $\gamma_\plan(\F)$) that there exists lanes $p\in [\lead(r',m),l)$ and $p'\in
      \Lanesout$ such that $(p,p') \in \poss_i^{l-1}(r',m)$ and $(p,p')$ and $(l,l')$
    are not compatible. Since $ \poss_i^{l-1}(r',m) = \poss_j^{l-1}(r',m)$, by
    Lemma~\ref{lem:sameposs}, and $\poss_j^{l-1}(r',m) = \poss_j(r',m)$, we have
    that $(l,l')= (\lane_j(r',m),\plan_j(r',m))$ is not compatible with
    $\poss_j(r',m)$. This means that $P^\plan_j(r_j(m)) = \noop$, so $\I,(r',m)
      \models \neg\going_j$, a contradiction. Thus, we must have $(l,l') \in
      \poss_i(r,m)$.
  \end{proof}
\end{toappendix}

\begin{propositionrep}
  $P^\plan$ implements $\kbp$ with respect to $\gamma_\plan(\F)$ for $\F \in
    \{CR, SO\}$.
\end{propositionrep}
\begin{proof}
  We show that for all points $(r,m)$ of $\I = \I_{P^\plan,\gamma_\plan(\F)}$,
  we have $P^\plan_i(r_i(m)) = \kbp^\I_i(r_i(m))$.
  If $\I,(r,m)\models \neg \front_i$, then we have $P^\plan_i(r_i(m)) = \noop = \kbp^\I_i(r_i(m))$.
  Suppose that $\I,(r,m)\models \front_i$.

  We first assume that $P^\plan_i(r_i(m)) = \go$ and show that $\kbp^\I_i(r_i(m))= \go$,
  By assumption,  $(\lane_i(r,m), \plan_i(r,m))$ is compatible with
  $\poss_i(r,m)$. By way of contradiction, suppose that $\kbp^\I_i(r_i(m)) \neq \go$.
  Then $\I,(r,m) \models \neg K_i(V_i')$, so there exists a point $(r',m) \sim_i (r,m)$
  and an agent $j$ with $\I,(r',m) \models j\in \GO$ and a move $(l,l')$ such that
  $\lane_j(r'm) = l$ and $\plan_j(r',m) = l'$ and $l \in [\lead,\lane_i(r',m))$
  and the move $(l,l')$ is not compatible with  $(\lane_i(r,m), \plan_i(r,m))$.
  By definition of $P^\plan$, the fact that $\I,(r',m) \models j\in \GO$  implies that
  $(l,l')$ is compatible with  $\poss_j(r',m) = \poss^{l-1}_j(r',m)$.
  By Lemma~\ref{lem:sameposs} and the fact that $(r,m) \sim_i (r',m)$,
  we have $\poss^{l-1}_j(r',m) = \poss^{l-1}_i(r',m)$. Hence $(l,l') \in \poss_i(r,m)$.
  But this means that  $(\lane_i(r,m), \plan_i(r,m))$ is not compatible with
  $\poss_i(r,m)$, hence  $P^\plan_i(r_i(m)) \neq \go$, a contradiction.

  For the converse, we suppose that $\kbp^\I_i(r_i(m))= \go$ and $P^\plan_i(r_i(m)) \neq \go$,
  and again derive a contradiction.  From $P^\plan_i(r_i(m)) \neq \go$,
  we have that $(\lane_i(r,m),\plan_i(r,m))$ is incompatible with some move
  $(l,l')\in \poss_i(r,m)$, for which $l \in [\lead,\lane_i(r,m))$.  By Lemma~\ref{lem:Pplan-poss},
  there exists a point $(r',m) \sim_i(r,m)$ with
  $\I,(r',m) \models  \exists j\in \Agents( \front_j \land \lane_j = l \land \plan_j = l' \land j\in \GO))$.
  Because $\kbp^\I_i(r_i(m))= \go$, we have $\I,(r,m) \models K_i(V'_i)$,
  hence $\I,(r',m) \models V'_i$. It follows that $(l,l')$ is compatible with $(\lane_i(r',m), \plan_i(r',m))$.
  But  $(\lane_i(r',m), \plan_i(r',m)) = (\lane_i(r,m),\plan_i(r,m))$, so we have a contradiction.
\end{proof}

Again, by Proposition~\ref{prop:optimal}, it follows that the intersection protocol $P^\plan$ is lexicographically optimal with respect to the contexts $\gamma_\plan(\F)$ for
for $\F \in \{CR, SO\}$.

\section{Discussion} \label{sec:discussion}
We introduced the \emph{intersection problem}, identified the appropriate
notion of optimality called \emph{lexicographical optimality}, and
designed protocols that are optimal in a variety of contexts.
A knowledge-based analysis and the use of \emph{intersection policies}
were crucial in this process.

Previous work has considered many models ranging from computing
individual trajectories of vehicles to relying on centralized schedulers
\cite{CE16}.
In \cite{ssp17,ciscav21}, a four-way intersection is considered
in a context with failures.
\cite{ffcvt10,ZSGPB18} consider \emph{virtual traffic lights}; the
approach is evaluated using a large-scale simulation.
\cite{fvpbk13} solves the same
problem
probabilistically, in contexts with failures.
Work in the control theory literature has focused on vehicle dynamics when
going through an intersection \cite{HCCV13} to avoid collision.
Efforts have also been made to
build distributed intersection management systems through V2V
communication \cite{CKSC13}.

While there has been considerable effort in designing protocols for specific
intersections or designing architectures for intersection management systems, we
aim to develop a context- and architecture-independent approach.
Our goal in this paper is to lay the theoretical foundations of optimal
intersection protocol design in a variety of contexts, including contexts
with failures. We do so abstractly by defining the model to capture
any intersection topology with minimal requirements on V2V communication range.
While the protocols we design do not require sensors such as
lidar and radar, the use of a knowledge-based program $\kbp$
provides a direct method to develop optimal implementations
in contexts with extra sensors.

The problem we study in this paper can be viewed as a generalization of the
classical problem of \emph{mutual exclusion}, which requires that two distinct
agents are not simultaneously in a \emph{critical section} of their code.
Indeed, a variant of mutual exclusion called \emph{group mutual
  exclusion}
\cite{Joung00} is strictly weaker
than the intersection problem.
In group mutual exclusion, each process is assigned a session when entering the
critical section and processes are allowed to enter the critical section
simultaneously provided that they share the same session. If agents form an
equivalence relation based on their
move
compatibility according to $\noconflict$, we
can identify each equivalence class to be in the same session and think of the
intersection as the critical section.
However, our setting differs in some critical ways:
\begin{itemize}
  \item Intersections often have an $\noconflict$ relation that is not an
        equivalence
        relation. For instance, the fact that agents' moves conflict in lanes
        A-B and in lanes B-C does
        not imply that their moves in lanes A and C conflict
        (e.g., if agents want to move straight in a
        four-way intersection with two lanes in each direction).
  \item We take the set $\Agents$ of agents  to be unbounded, while
        group mutual exclusion (and equivalent problems such as
        \emph{room synchronization} \cite{BCG01}) consider a bounded number of agents.
  \item Our agents arrive according to a (possibly infinite)
        schedule determined
        by the
        adversary.
  \item To the best of our knowledge, fault-tolerance has not been
        considered
        in the group mutual-exclusion setting.
\end{itemize}
The mutual-exclusion problem is generally studied with respect to an
interleaving model of asynchronous computation, but as Lamport \cite{Lamport74a}
noted, this model is not physically realistic, and already builds in a
notion of mutual exclusion between the actions of distinct agents. The
\emph{Bakery} mutual-exclusion protocol \cite{Lamport74a}
is correct with respect to models allowing
simultaneous read and write operations.
Moses and Patkin \cite{MosesP18} develop an improvement of Lamport's Bakery
algorithm for the mutual-exclusion problem
using a knowledge-based analysis,
noting that there are situations in
which Lamport's protocol could enter the critical section, but fails to do so.
A weaker knowledge-based condition for mutual exclusion is used
by Bonollo et al.~ \cite{BMS01}; it states that an agent $i$ may enter
its critical section when it knows that no other agent will enter its critical
section until agent $i$ has exited from its critical section.
Clearly these knowledge-based approaches are similar in spirit to ours.
We hope to study the exact relationship
between these problems in the near future.

There are several directions that we hope to explore in the future.
One involves extending the current results to contexts with stronger
adversaries and
evaluating implementations of $\kbp$
in other contexts.
Another is considering strategic agents, who may deviate from a
protocol to cross the intersection earlier.

\bibliography{joe,z,refs}

\end{document}